\documentclass[acmsmall]{ec21acm}
\usepackage{booktabs} 
\usepackage[ruled]{algorithm2e} 

\SetAlFnt{\small}
\SetAlCapFnt{\small}
\SetAlCapNameFnt{\small}
\SetAlCapHSkip{0pt}
\IncMargin{-\parindent}
\usepackage{subfigure}

\hypersetup{colorlinks=true,linkcolor=blue,citecolor=blue}
\usepackage[shortlabels]{enumitem}
\renewcommand{\vec}[1]{\mathbf{#1}}

\newcommand{\eps}{\epsilon}

 \newcommand{\bij}{b_{i,j}} 
\newcommand{\xij}{x_{i,j}}
\newcommand{\aij}{a_{i,j}} 
\newcommand{\fisher}{\operatorname{\mathsf{Fisher}}}

\definecolor{mygreen}{rgb}{0.0, 0.55, 0.0}
\definecolor{amethyst}{rgb}{0.6, 0.4, 0.8}
\definecolor{blue(pigment)}{rgb}{0.2, 0.2, 0.6}
\definecolor{blue-violet}{rgb}{0.54, 0.17, 0.89}
\definecolor{blush}{rgb}{0.87, 0.36, 0.51}

\newtheorem{remark}{Remark}

\newtheorem{lemma}{Lemma}
\newtheorem{definition}{Definition}

\newtheorem{example}{Example}
\newtheorem{theorem}{Theorem}
\newtheorem{corollary}{Corollary}

\newtheoremstyle{case}{}{}{}{}{}{:}{ }{}
\theoremstyle{case}

\setcitestyle{acmnumeric}

\AtBeginDocument{%
  \providecommand\BibTeX{{%
    \normalfont B\kern-0.5em{\scshape i\kern-0.25em b}\kern-0.8em\TeX}}}

\copyrightyear{2021}
\acmYear{2021}
\setcopyright{acmcopyright}\acmConference[EC '21]{Proceedings of the 22nd ACM Conference on Economics and Computation}{July 18--23, 2021}{Budapest, Hungary}
\acmBooktitle{Proceedings of the 22nd ACM Conference on Economics and Computation (EC '21), July 18--23, 2021, Budapest, Hungary}
\acmPrice{15.00}
\acmDOI{10.1145/3465456.3467644}
\acmISBN{978-1-4503-8554-1/21/07}
\begin{document}

\title{Proportional Dynamics in Exchange Economies}

\author{Simina Br\^anzei}
\affiliation{%
  \institution{Purdue University}
  \city{West Lafayette}
  \country{United States}}
\email{simina@purdue.edu}

\author{Nikhil Devanur}
\affiliation{%
  \institution{Amazon}
  \city{Seattle}
  \country{United States}
}

\author{Yuval Rabani}
\affiliation{%
 \institution{Hebrew University of Jerusalem}
 \city{Jerusalem}
 \country{Israel}}

\renewcommand{\shortauthors}{Br\^anzei, Devanur, and Rabani}

\begin{abstract}
	We study the proportional dynamics in exchange economies, where each player starts with some amount of money and a good. Every day, players bring one unit of their good and submit bids on goods they like, each good gets allocated in proportion to the bid amounts, and each seller collects the bids received. Then every player updates their bids proportionally to the contribution of each good in their utility.
\medskip

This dynamic models a process of learning how to bid and has been studied in a series of papers on Fisher and production markets, but not in exchange economies. 
Our main results are as follows:
\begin{enumerate}
	\item For all linear utilities, the dynamic converges to market equilibrium utilities and allocations, while the bids and prices may cycle. We give a combinatorial characterization of limit cycles for prices and bids. 
	\item We introduce a lazy version of the dynamic, where players may save money for later, and show this converges in everything: utilities, allocations, and prices.
\end{enumerate}

\medskip

This answers an open  question about exchange markets with linear utilities, where t\^atonnement does not converge to market equilibria, and no natural process leading to equilibria was known for all additive utilities. We also note this dynamics represents a process where the players exchange goods throughout time (in out-of-equilibrium states), while t\^atonnement only explains how exchange happens in the limit.
\end{abstract}




\maketitle

\section{Introduction}

Market \emph{dynamics} have been an integral part of general equilibrium theory since its inception. 
The introduction of general equilibrium theory by Walras~\cite{walras1896elements} was accompanied 
by the idea of the \emph{t\^atonnement} process. Fisher designed in 1891 a device to compute an equilibrium
(Brainard and Scarf~\cite{BS05}). The most popular interpretation of t\^atonnement, however, is as \emph{fictitious} play. 
An auctioneer, playing the role of \emph{the invisible hand}, calls out prices to which the agents respond 
with their demand, then the auctioneer adjusts the prices, and the process repeats until the excess demand 
is zero. At this point exchange actually happens, at equilibrium prices. This, however, is not necessarily how 
actual markets function in practice. Moreover, general equilibrium theory itself suffers from a lack of 
a \emph{descriptive model} of out of equilibrium exchange: what happens when the excess demand is positive? 
Further, the causal linkage between demand and prices is unspecified. Demand is a response to the price 
as well as the price is a response to the demand. (See~\cite{Fisher83} for some disequilibrium extensions;
there is no widely agreed upon model.) The question is fundamentally both algorithmic and economic.
Algorithmically, the question is about an effective and efficient locally controlled network process that
computes an equilibrium. Economically, the question is about an incentives-motivated multi-agent process
that converges to equilibrium.

\medskip

Shapley and Shubik~\cite{shapley1977trade} sought to address these issues via the \emph{trading post} mechanism. 
This is first of all a descriptive model that specifies concrete outcomes as a result of player strategies, and 
therefore it can be viewed as a non-cooperative game. Prices are a result of strategic actions; the higher the 
demand for a good, the higher its price, and vice versa. It requires that all trade be \emph{monetary} and that 
the players pay \emph{cash in advance}.\footnote{The assumption is that there is one special commodity 
	used as the means of payment, which is called \emph{cash} or \emph{money}. It may or may not have an intrinsic utility of its 
	own.} Each player submits a cash bid on each \emph{good}. The goods are then distributed \emph{in proportion 
	to the bids}, and the per-unit price of a good is set to be the sum of bids on that good (this implies that the bids 
of a player should add up to at most the cash at hand). The same mechanism has been rediscovered multiple 
times; in particular it has been proposed for sharing resources in computer networks (e.g. Kelly~\cite{kelly1997charging})
and computer systems (e.g. Feldman, Lai, and Zhang~\cite{feldman2005price}).

\medskip

This still leaves open the question of dynamics: (how) do players reach a market equilibrium in the trading post 
mechanism? The predominant answer to this in the last decade or so has been the {proportional response} 
dynamic (see Zhang~\cite{Zhang11}), where buyers iteratively update their bid on each good \emph{in proportion to the utility} they received 
from that good in the previous iteration. For the case of linear utilities, this implies that the ratio of bids in successive 
iterations is proportional to the \emph{bang-per-buck} for that good, in contrast to the \emph{best response} which 
distributes \emph{all} the cash among the goods with the \emph{highest} bang-per-buck. Such an update is similar 
in spirit to the multiplicative weights update algorithms used in online learning (also to proportional t\^atonnement), 
except that there are no parameters such as the \emph{step size} to tune carefully! It still magically seems to work. Wu and Zhang~\cite{wu2007proportional} studied a dynamic without money, in a special type of exchange market, where each good has a common value (i.e. the value of any good $j$ is the same for every player $i$: $v_{i,j} = v_j$), and showed it converges to market equilibria.

\subsection{Our results}
This brings us to the topic of this paper, the study of proportional response dynamics in exchange markets. The
players are both buyers and sellers, and the exchange at each step is fueled by the players revenue from the previous
round, and a fresh batch of goods of fixed quantities at the hands of the players. We focus primarily on \emph{linear utilities}. 
Linear utilities are a widely studied type of utility that are also of interest because the process of t\^atonnement 
is not well defined in markets with such utilities. T\^atonnement adjusts the prices based on the excess demand for a good, 
but with linear utilities the demand is a set function. Also, the demand is discontinuous: a small change in price can lead to 
a large change in demand. This makes t\^atonnement especially unsuited as a process describing market dynamics for linear 
utilities. In contrast, linear utilities pose no such problem for proportional response dynamics. For complements, especially 
in the extreme case of Leontief utilities, there is scant hope for fast convergence since computationally the problem of finding 
an equilibrium is PPAD-hard (see Codenotti, Saberi, Varadarajan, and Ye~\cite{codenotti2006leontief}). However this does not rule out the existence of a slowly converging 
process.

\medskip

The dearth of convergence results in this setting is not for the lack of trying (see, e.g., gradient descent based algorithms~\cite{CLL19}). The difficulty might be attributed to 
the fact that the dynamics can cycle! Consider the scenario where there are two sets of players such that each set buys all 
its goods from the other set. Suppose that the sets start with widely unequal amounts of cash. Then in each iteration, the 
total amount of cash of each set moves to the other side, forever. However there is still hope: the allocations and the utilities 
of the players could converge to those of a market equilibrium, even if the prices oscillate. (In fact, in the above example
the relative prices in each set may converge, even though the price scales alternate between the sides.) 

\medskip

We study the exchange economy model, where each player comes to the market with a unit endowment of an 
exclusive good. For linear utilities, this model is without loss of generality for the equilibrium computation problem, as there is a reduction from the general model, where each player can bring multiple goods, to the one where each player brings only one good. We 
The equilibrium utilities are unique, while the equilibrium allocation may not be. 
Our main results are as follows: 
\begin{enumerate}
	\item We show that the kind of cycling described above is essentially the only one possible. We characterize the 
	\emph{limit cycles} of the dynamics as follows: there are equivalence classes of players such that 
	within each class the ratio of price to equilibrium price is a constant. Further, the classes form a cycle, where 
	the players in each class only buy goods from the players in the next class in the cycle. The allocations and utilities 
	correspond to a market equilibrium, and remain invariant all along the limit cycle. (Theorem~\ref{thm: limit alloc_cycling_bids})
	
	\item We show that the allocation and hence the utilities converge (Theorem~\ref{thm:nonlazy_allocation}). The result in the previous item only shows that 
	the limit set in the allocation space is the set of equilibrium allocations. Convergence to this set is implied.
	\item We introduce a \emph{lazy} version~\footnote{The name is motivated by lazy random walks.} of the proportional response dynamics,
	where players saves a certain fraction of their cash at hand for future rounds (Definition~\ref{def:lazypr}). The fraction can be different for each 
	player, as long as it is in $(0,1)$, but it doesn't change over time. We show that for this version, there is no cycling: 
	prices, allocation and utilities all converge to a market equilibrium (Theorem~\ref{thm:lazy_convergence}). 
	\item We also prove an ergodic rate of convergence of $O(1/T)$ for the utilities, with respect to the Eisenberg-Gale objective function, which captures the market equilibrium allocations and utilities.\footnote{This objective is defined for Fisher markets with fixed budgets. We use this objective with the budgets set to equilibrium incomes.} We note that our simulations indicate that the last iterate also converges, and that the convergence rate is much faster than proved.

\end{enumerate}

\subsection{Previous work}
Most of the results for the convergence of the proportional response dynamics to date have been for \emph{Fisher} 
markets, where the players act as buyers, and each step is fueled by a fixed income of each player and a fresh
batch of goods. Zhang~\cite{Zhang11} shows convergence of proportional response dynamics to the market 
equilibrium for \emph{Constant Elasticity of Substitution} (CES) utilities in the \emph{substitutes} regime. 
Birnbaum, Devanur, and Xiao~\cite{BDX11} interpret proportional as \emph{mirror descent} on the convex program
of Shmyrev~\cite{shmyrev2009algorithm} that captures equilibria in Fisher markets with \emph{linear} utilities, and also 
extends it to some other markets. Cheung, Cole, and Tao~\cite{CCT18} extend the approach in~\cite{BDX11} to show that 
proportional response converges for the entire range of CES utilities including \emph{complements}, with linear 
utilities on one extreme and \emph{Leontief} utilities on the other extreme. Cheung, Hoefer, and Nakhe~\cite{CHN19} show that 
the dynamics stays close to equilibrium even when the market parameters are changing slowly over time, once again 
for CES utilities.

\medskip

Wu and Zhang~\cite{wu2007proportional} also study a dynamics in an exchange setting. The Wu-Zhang dynamic is different from the one we study because it does not use money and converges only in the special case of the exchange market where each good has a common value (i.e. where the value of each player $i$ for each good $J$ is $v_{i,j} = v_j$). The Wu-Zhang dynamics is motivated
by application of exchange from networking, where it is preferable to not use money. We note that the dynamics in our paper is the correct generalization of proportional response as it was studied for Fisher markets (see Zhang \cite{Zhang11} and Birnaum-Devanur-Xiao\cite{BDX11}).
We also show that the dynamics from ~\cite{wu2007proportional} and our proportional response dynamics are functionally different, 
i.e. they have different trajectories in terms of allocations and utilities, even in the special case of the market where goods have common values and given the same starting configurations.
In fact, when the goods do not have a common value (i.e. the players can value different goods differently), we find an example of a market linear where the Wu-Zhang dynamics cycles (see Appendix~\ref{apx:comparison}).


\medskip

Branzei, Mehta, and Nisan~\cite{BMN18} generalized the definition of proportional response from Fisher markets to an exchange setting with production. The definition of proportional response in~\cite{BMN18} is the same as the one we use, except the amounts are fixed over time in our model. On the other hand, in the production market, players make new goods from the ones 
they acquire through the trading post mechanism, and the new goods are sold in the next iteration. There, the proportional response dynamics  
leads to \emph{growth} of the market, i.e. the amount of goods produced
grow over time, but also to growing inequality between the players on the most efficient 
production cycles and the rest. 

The study of convergence of t\^atonnement goes back at least as far as Arrow, Block, and Hurwicz~\cite{arrow1959stability}, 
which was soon followed by examples of cycling (see Scarf~\cite{scarf1960some} and Gale~\cite{Gale63}). For markets with weak gross
substitutes utilities (WGS), a polynomial time convergence of a discrete time process was shown by 
Codenotti, McCune, and Varadarajan~\cite{codenotti2005market}, and Cole and Fleischer~\cite{cole2008fast} showed fast convergence not just 
for static markets, but also ``ongoing'' markets. (See also Fleischer, Garg, Kapoor, Khandekar, and Saberi~\cite{FGKKS08}.) This was followed up by similar analysis 
for some markets with complementarities (Cheung, Cole, and Rastogi~\cite{cheung2012tatonnement}; Cheung and Cole~\cite{cheung2014amortized}; Avigdor, Rabani, and Yadgar~\cite{avigdor2014convergence}),
then for all ``Eisenberg-Gale'' markets in the Fisher model~\cite{cheung2019tatonnement}. Some of these analyses
apply to ongoing and/or asynchronous settings.

\medskip

Several approaches have been explored for the (centralized) computation of equilibria in a linear exchange market: 
the ellipsoid method (Jain~\cite{Jain}; Codenotti, Pemmaraju, and Varadarajan~\cite{codenotti2005polynomial}), interior point algorithms (Ye~\cite{ye2008path}), 
combinatorial flow based methods (Jain, Mahdian, and Saberi~\cite{jain2003approximating}; Devanur and Vazirani~ \cite{devanur2003improved}; Duan and Mehlhorn~\cite{duan2015combinatorial}; Duan, Garg, and Mehlhorn~\cite{duan2016improved}).
Extending the Fisher market version of~\cite{devanur2008market}) recently led to a strongly polynomial time 
algorithm (Garg and V{\'e}gh~\cite{garg2019strongly}). Other approaches include auction based algorithms~\cite{garg2006auction,bei2019ascending}, 
cell decomposition~\cite{deng2003complexity,devanurKannan2008market}, complementary pivoting~\cite{garg2015complementary}, 
and computational versions of Sperner's lemma~\cite{echenique2011finding,scarf1977computation}.
On the other hand, computing an equilibrium with even the simplest kind of complementarities is 
PPAD-hard~\cite{codenotti2006leontief,chen2009settling}. 
\medskip

There has been extensive work on understanding dynamics in games and auction settings under various behavioural models of the agents, such as best-response dynamics, multiplicative weight updates, fictitious play (e.g., \cite{FS99,KPT09,DDK15,MPP15,PP16,RST17,DS16,HKMN11,PpadP16,LST16}) and best response processes and other dynamics of learning how to bid in market settings \cite{CD11,nisan2011best,bhawalkar2011welfare,lucier2010price,BBN17,DK17,CDEG+14,BFR19}). 
In the former the focus has been on convergence to an equilibrium, preferable Nash, and if not then (coarse) correlated equilibria, and the rate of covergence. In the latter the focus has been on either convergence points and their quality (price-of-anarchy), or dynamic mechanisms such as ascending price auctions to reach efficient allocations (e.g. the Ausubel auction \cite{Ausubel04}).

\subsection{Difficulties and techniques}
The strongest convergence results for proportional response dynamics in Fisher markets are achieved via the mirror descent 
interpretation on suitable convex programs~\cite{BDX11,CCT18}. Devanur, Garg, and V{\'e}gh~\cite{devanur2016rational} show a similar 
(but more complicated) convex program for linear utilities in exchange markets. It is therefore tempting to conjecture that a 
similar mirror descent interpretation would extend to exchange markets as well, but unfortunately this doesn't seem to be the 
case. There are many difficulties, but the easiest to explain is the following. In the Fisher case, the proportional response bids 
are by definition in the feasible region of the convex program, which asks that the price of a good equal the total bids placed 
on it, and that the budget of a player equal the total of his own bids. In the exchange case, the price of a good is still the total 
bid placed on it by definition, but the total bid a player issues equals his earnings from the previous iteration, which may
differ from the total bid placed on its good in the current iteration. The convex program requires these equality constraints, 
so the bids don't stay inside the feasible region as in the Fisher case. 
\medskip

We use the KL divergence between equilibrium bids and the current bids as a Lyapunov function; this divergence was also used in \cite{Zhang11} when analyzing Fisher markets.
The convergence of utilities is the easiest, and follows almost exactly the analysis for Fisher markets. Beyond utilities, the 
cycling of bids presents more difficulties. We characterize the limit cycles by considering the zero set of a certain set of equations. 
We argue that their structure is like that of the price ratios in the limit cycles described above. The convergence of allocation 
follows from showing that the KL divergence between the (equilibrium and current) bids can be decomposed into a positive 
linear combination of the KL divergence between the prices and the KL divergences between the  allocations for each good. 
This also implies that the KL divergence between equilibrium prices and prices on a limit cycle must be an invariant. 

\medskip

For the lazy version we show that adding a suitably weighted KL divergence between equlibrium prices and current budgets 
to the Lyapunov function does the trick. This collapses the limit cycles so that the limit set becomes just the equilibria. This then gives 
us that the Lyapunov function must go to zero, which implies convergence of prices as well.

\subsection{Organization of the paper}
In Section~\ref{sec:prelims}, we give the definitions of the two variants of the proportional response dynamics, 
and market equilibria and state some useful properties. 
We also show numeric examples of cycling for the non-lazy version. 
Section~\ref{sec:lazy} defines a Lyapunov function for the dynamics, which shows convergence of utilities. 
We also show convergence of allocation and prices for the strictly lazy version. 
In Section~\ref{sec:cycling}, we characterize limit cycles, and show convergence of allocation for the non-lazy version. 
Section \ref{apx:comparison} shows a comparison with the tit-for-tat dynamic, including an example where tit-for-tat cycles. 
We can conclude with a discussion in Section 6.

\section{Preliminaries}
\label{sec:prelims} 

There are $n$ agents, each of which has one unit of an eponymous good. 
The goods are divisible. 
The agents have linear utilities given by a matrix $A = \{a_{i,j}\}$, where $a_{i,j} \geq 0$ is the valuation 
of agent $i$ for one unit of the good owned by agent $j$. 
We assume that every agent has a certain quantity of a numeraire to start with, which we call money. 
All the prices will be determined in terms of this numeraire. 
The agents don't have any utility for the money; 
it just facilitates exchange of goods.

The assumption on the valuations $a_{i,j}$ will be that a market equilibrium exists for the induced exchange market. An equilibrium does not exist when some agent has no edges; in some sense such agents do not participate in the economy and will not be considered in the dynamics.

\subsection{Proportional Response Dynamics in Exchange Economies:}
The proportional response dynamics describes a process in which the players come to the market every day with one unit of good and some budget, which is split into bids. The players bid on the goods, then the seller of each good allocates it in proportion to the bid amounts and collects the money from selling, which becomes its budget in the next round. Finally the players update their bids in proportion to the contribution of each good to their utility.

\begin{definition}[Proportional Response Dynamics]\label{def: PRD}
	The initial bids of player $i$ are $b_{i,j}(0)$, which are non-zero whenever $a_{i,j} > 0$. Then, at each time $t$, the following steps occur:
	
	\leftskip1.5em	\paragraph{Exchange of goods:} Each player $i$ brings one unit of its good and submits bids $b_{i,j}(t)$. The player receives an amount $x_{i,j}(t)$ of each good $j$, where 
	\begin{equation*}
	x_{i,j}(t) = \left\{
	\begin{array}{ll}
	\frac{b_{i,j}(t)}{\sum_{k=1}^n b_{k,j}(t)}, & \mbox{if} \; b_{i,j}(t) > 0\\
	0, & \mbox{otherwise}
	\end{array}
	\right.
	\end{equation*}
    \hspace{2.5mm} Each player $i$ computes the utility for the bundle acquired:
	$
	u_i(t) = \sum_{j=1}^{n} a_{i,j} \cdot x_{i,j}(t)\,.
	$
	\leftskip1.5em	\paragraph{Bid update:} Each player $i$ collects the money made from selling: $B_i(t+1) = \sum_{k=1}^n b_{k,i}(t)$ and updates his bids proportionally to the contribution of each good in his utility:
	$$\textstyle
	b_{i,j}(t+1) = \left( \frac{a_{i,j} \cdot x_{i,j}(t)}{u_{i}(t)} \right) \cdot B_i(t+1) \,.
	$$	
\end{definition}

\noindent \textbf{Normalization}: Without loss of generality, we can assume the total amount of money in the economy is $1$.

\begin{remark} The sum of bids on a good can be seen as its \emph{price}, so we will write $p_i(t) = \sum_{k=1}^n b_{k,i}(t)\,.$
\end{remark}

The next figure shows an instance with twelve players. The allocations, utilities, and prices oscillate initially but stabilize later.

\begin{figure}[H]
	\centering
	\subfigure[Allocations $x_{i,j}(t)$ over time.]
	{
		\includegraphics[scale=0.33]{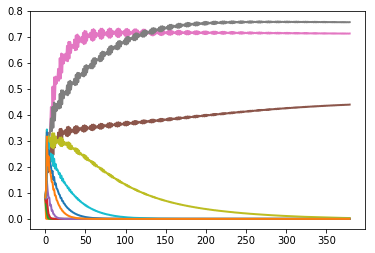}
	}
	\subfigure[Utilities $u_i(t)$ over time.]
	{
		\includegraphics[scale = 0.33]{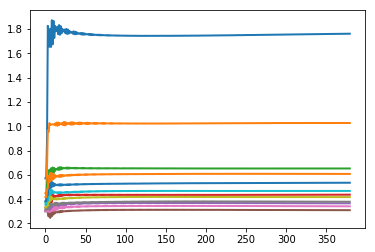}
	}
	\subfigure[Bids $b_{i,j}(t)$ over time.]
	{
		\includegraphics[scale = 0.33]{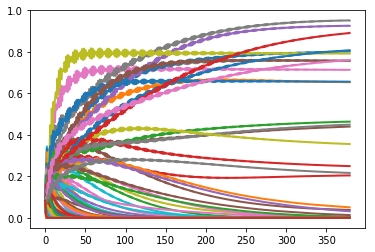}
	}
	\caption{Proportional response dynamics in an economy with $n=12$ players over $T = 380$ rounds: allocations, utilities, and bids.}
\end{figure}

Figure \ref{fig:cycling} shows an example trajectory in a ten player economy. In this example it can be seen 
that the allocations and utilities converge while the prices cycle. 

\begin{figure}[H]
	\centering
	\subfigure[Allocations $x_{i,j}(t)$ over time.]
	{
		\includegraphics[scale=0.331]{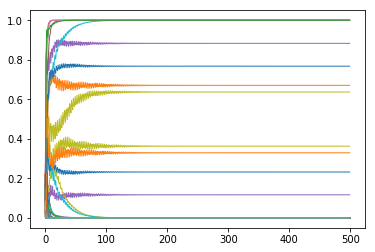}
		\label{fig:allocationscycling}
	}
	\subfigure[Utilities $u_i(t)$ over time.]
	{
		\includegraphics[scale = 0.331]{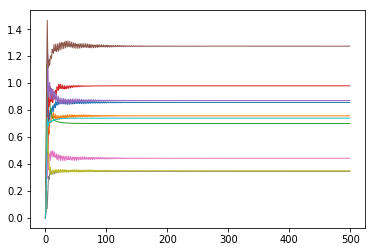}
		\label{fig:utilitiescycling}
	}
	\subfigure[Prices $p_i(t)$ over time.]
	{
		\includegraphics[scale = 0.331]{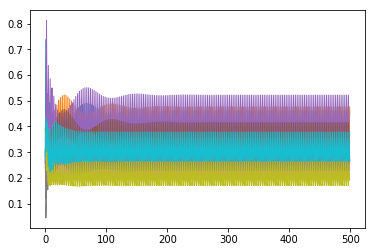}
		\label{fig:pricescycling}
	}
	\caption{Proportional response dynamics in an economy with $n=10$ players over $T = 500$ rounds: allocations, utilities, and prices. The set of players is divided in three groups: $\left(C_1, C_2, C_3\right)$, where $C_4 = C_1$, $C_1 = \{1, 2, 3\}$, $C_2 = \{4,5,6,7\}$, and $C_3 = \{8, 9, 10\}$. The valuations are such  that the players in group $C_i$ only value the goods of players in group $C_{i+1}$.}
	\label{fig:cycling}
\end{figure}

\subsection{Lazy Proportional Response Dynamics in Exchange Economies:}

Our first contribution is to define a more general framework for the proportional response dynamics, which can be seen as a lazy version of the dynamic. Each player will spend only some fraction $\alpha_i \in (0,1]$ of its total money in each round, while saving the remaining fraction of $1- \alpha_i$ in the bank. Then in the next round, the player collects the money it made from selling and takes out the money from the bank, which sum up to its total amount of money. Then again the player spends a fraction $\alpha_i$ of its total money, while saving the remainder of $1 - \alpha_i$ in the bank.

\medskip

Formally, we have the following definition.

\begin{definition}[Lazy Proportional Response Dynamics]
	\label{def:lazypr} 
	Initially each player $i$ has some amount of money $m_i(0)$. The player splits a fraction $\alpha_i \in (0, 1]$ of the money into initial bids $b_{i,j}(0)$, which are non-zero whenever $a_{i,j} > 0$. I.e.,
	the budget for spending at time $0$, which is $B_i(0)\triangleq \alpha_i \cdot m_i(0)$, is split into bids $b_{i,j}(0)$ satisfying 
	$\sum_{j} b_{i,j}(0) = \alpha_i \cdot m_i(0)$, and the player saves the remaining portion of money $(1-\alpha_i)\cdot m_i(0)$ in the bank.
	
	\vspace{2mm}
	At each time $t$, the next steps take place:
	
	\leftskip1.5em	\paragraph{{Exchange of goods:}} Each player $i$ brings one unit of its good and submits bids $b_{i,j}(t)$. The player receives an amount $x_{i,j}(t)$ of each good $j$, where 
	\begin{equation*}
	x_{i,j}(t) = \left\{
	\begin{array}{ll}
	\frac{b_{i,j}(t)}{\sum_{k=1}^n b_{k,j}(t)}, & \mbox{if} \; b_{i,j}(t) > 0\\
	0, & \mbox{otherwise}
	\end{array}
	\right.
	\end{equation*}
	\hspace{2.5mm} Each player $i$ computes the utility for the bundle acquired:
	$
	u_i(t) = \sum_{j=1}^{n} a_{i,j} \cdot x_{i,j}(t).
	$
	\leftskip1.5em	\paragraph{Bid update:} Each player $i$ collects the money made from selling: $p_i(t) = \sum_{k=1}^n b_{k,i}(t)$.
	
	The total money of player $i$ is the price of the good sold plus the money saved: 
	$m_i(t+1) = p_i(t) + (1- \alpha_i) B_i(t)/\alpha_i$. This gets split again into a fraction of $(1- \alpha_i)$, saved in the bank, and a fraction $\alpha_i$ that will be spent in the next round: $\textstyle
	B_i(t+1) = \alpha_i \cdot m_i(t+1) = \alpha_i \cdot p_i(t) + (1-\alpha_i) B_i(t).$
	
	The player updates his bids proportionally to the contribution of each good in his utility:
	$$\textstyle
	b_{i,j}(t+1) = \left( \frac{a_{i,j} \cdot x_{i,j}(t)}{u_{i}(t)} \right) \cdot B_i(t+1) 
	$$	
\end{definition}

Notice that the special case of this dynamic where $\alpha_i = 1$ for all agents $i$ is simply the proportional response dynamic of
Definition~\ref{def: PRD}. 

\medskip

In Figure~\ref{fig:noncycling_lazy} we show the dynamics for the same valuations and initial bids as in Figure \ref{fig:cycling}, but where the players  save half of their money in the bank at each time unit. 

\begin{figure}[H]
	\centering
	\subfigure[Allocations $x_{i,j}(t)$ over time.]
	{
		\includegraphics[scale=0.331]{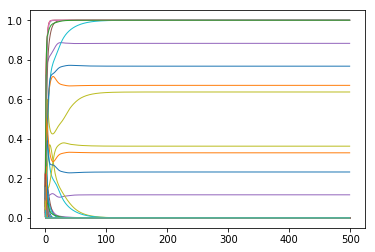}
		\label{fig:allocations_noncycling_lazy}
	}
	\subfigure[Utilities $u_i(t)$ over time.]
	{
		\includegraphics[scale = 0.331]{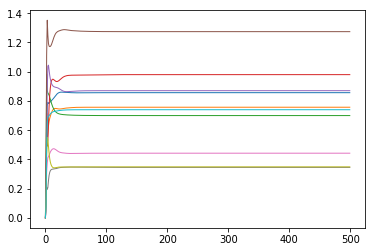}
		\label{fig:utilities_noncycling_lazy}
	}
	\subfigure[Prices $p_i(t)$ over time.]
	{
		\includegraphics[scale = 0.331]{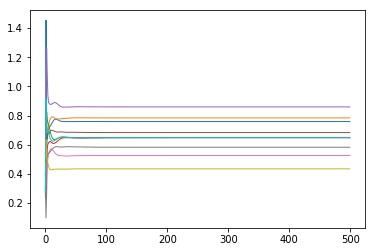}
		\label{fig:prices_noncycling_lazy}
	}
	\caption{Lazy version of the dynamics where the players have the same valuations and initial budgets and bids as in Figure \ref{fig:cycling}, but they also have a matching amount saved in the bank, and in each round they save $50\%$ of their money.}
	\label{fig:noncycling_lazy}
\end{figure}

\subsection{Market Equilibria}
We review the definition of market equilibria and some useful properties. 
We assume that for each good $j$ there is at least one player $i$ such that $\aij > 0$. 
An equilibrium is given by a set of prices $p_i$ for each $i\in [n]$ 
and a set of allocations $\xij\ge 0$ for each pair $i,j\in[n]$ such that 
\medskip 

\begin{description}
	\item [Market clearing:] the goods are all sold, i.e.,  $\forall~j, \sum_i \xij = 1$. 
	\item [Optimal allocation:] each buyer gets an optimal bundle of goods, i.e.,  $\forall~i, \xij$s maximize the sum $\sum_j \aij \xij $ subject to the budget constraint $\sum_j p_j \xij \leq p_i$.  
\end{description}

\medskip 

It is known that equilibrium utilities are unique.
Equilibrium allocations and prices may not be unique,
but equilibrium allocations, equilibrium prices,  the set of equilibria $(\xij,\log p_j)$, 
and the set of equilibria $(\bij, p_j)$ where $\bij = \xij p_j$, 
all form convex sets \cite{gale1976linear,cornet1989linear,mertens2003limit,florig2004equilibrium,devanur2016rational}.

\medskip

We note the following condition, which is guaranteed to hold at any equilibrium in an exchange market ($*$s indicate equilibrium quantities):
\begin{align} \label{eq:condition}
\textstyle
\frac{u_i^*}{p_i^*} = \frac{a_{i,j}}{p_j^*}, \mbox{for all} \; i,j \; s.t. \; x_{i,j}^* > 0
\end{align}

\paragraph{Fisher Markets:}
A variant of this model is the Fisher market, where there is a distinction between buyers and sellers. 
There are $n$ players and $n$ goods, and the utilities are as before. 
In addition, each player $i$ comes to the market with a fixed budget $B_i$. 
The equilibrium conditions are the same as before, except that the budget constraint of player $i$ is that $\sum_j p_j \xij \leq B_i$.
The following convex program, called the \emph{Eisenberg-Gale} convex program, 
captures equilibria in the Fisher market with linear utilities: the set of optimal solutions to this program  
is equal to the set of equilibrium allocations and utilities~\cite{eisenberg1959consensus}:
\[ \textstyle   \max \; \; \sum_{i \in [n]} B_i \cdot \log (u_i) \text{ s.t.} \]
\[ \textstyle\forall~ i \in [n], \; u_i \leq \sum_j \aij \xij ,\]
\[\textstyle \forall~j \in [n], \; \sum_i \xij \leq 1, \]
\[ \forall~ i,j \in [n], \; \xij \geq 0.\]
Moreover, equilibrium utilities and prices are unique in the Fisher market~\cite{eisenberg1959consensus}. 

\medskip

We conclude this section by noting that the fixed points of the lazy proportional response dynamics are market equilibria. 
\begin{proposition} \label{prop:fixed_points}
	Suppose $\alpha_i \in (0, 1]$ for each player $i$. Then any fixed point of the lazy proportional response dynamics is a market equilibrium.
\end{proposition}
\begin{proof}
	Suppose $\vec{b}^*$ is a fixed point of the proportional response dynamics; let $x_{i,j}^*$ be the resulting fixed point allocation and $u_i^*$ the corresponding utilities. For each good $j$, let $p_j^* = \sum_{i=1}^n b^*_{i,j}$; this quantity can be interpreted as the price of the good at the fixed point.
	The budget update rule at the fixed point gives:
	$$B_i^* = \alpha_i \cdot p_i^* + (1-\alpha_i) B_i^* \iff B_i^* = p_i^*.$$
	
	Using the fact that $B_i^* = p_i^*$ we can write the bid update rule as:
	\begin{align}
	& b_{i,j}^{*} = \left(\frac{a_{i,j} \cdot x_{i,j}^*}{u_i^*} \right)\cdot B_i^*=  \left(\frac{a_{i,j} \cdot \left(\frac{b_{i,j}^{*}}{\sum_{k=1}^{n} b_{k,j}^{*}}\right)}{u_i^*}\right) \cdot B_i^* =
	\left(\frac{a_{i,j} \cdot \frac{b_{i,j}^{*}}{p_j^*}}{u_i^*}\right) \cdot p_i^*  \notag 
	\end{align}
	If $b_{i,j}^* = 0$ the identity holds trivially. For $b_{i,j}^* > 0$, the identity is equivalent to 
	$u_i^*/p_i = a_{i,j}/p_j$.
	This condition is the same as the market equilibrium condition for all strictly positive bids in the exchange economy, and so every fixed point is a market equilibrium.
\end{proof}


\section{Convergence of lazy proportional dynamic} 
\label{sec:lazy} 

In this section we study the dynamic and show convergence of utilities for any combination of the values $\alpha_i \in (0,1]$.

\subsection{Convergence of Utilities}

In this section we show that the utilities of the players converge for any valuations, initial configuration of the bids $b_{i,j}(0)$, and savings fractions $\alpha_i \in (0,1]$ of the players.
We let $u_i^*$ denote the (unique) equilibrium utilities. 
\begin{theorem} \label{thm:util_me}
	For any initial non-zero bids, the utilities of the players in the  dynamic converge to the market equilibrium utilities $u_i^*$, for any savings $\alpha_i$; that is, $\lim_{t \to \infty} u_i(t) = u_i^*$. 
\end{theorem}
The high level idea is to show that a Lyapunov function for the dynamics is essentially the Kullback-Leibler (KL) divergence between the vector with the bids and prices at a market equilibrium, and bids and budgets for the dynamic. 
We let $p_i^*$ and $\xij^*$  denote some equilibrium price and allocation resp., and let $\bij^* = p_j^* \cdot \xij^*$. 
For each $i,j \in N$, define
\begin{align}
z_{i,j}(t) = 
\begin{cases} 
\left(\frac{b_{i,j}^*}{b_{i,j}(t)} \right)^{b_{i,j}^*} & if\; b_{i,j}^* > 0, b_{i,j}(t) > 0\\
1 & otherwise
\end{cases}
\end{align}
and 
\begin{align}
w_{i}(t) = 
\begin{cases} 
\left(\frac{p_i^*}{B_i(t)} \right)^{p_i^* \cdot \left(\frac{1-\alpha_i}{\alpha_i}\right)} & if\; p_i^*>0, B_i(t) > 0 \\
1 & otherwise
\end{cases}
\end{align}
Then our Lyapunov function is $f(t) = \prod_{i,j \in [n]} z_{i,j}(t) \cdot \prod_{i \in [n]} w_i(t)$. We will override notation and write $f(t)$ or $f(b(t))$, depending on whether it is necessary to emphasize the bids.

The key fact we use about $f(t)$ is an iterative formula 
relating $f(t+1)$ to $f(t)$. 
We first define 
the functions $h,g : \mathbb{N} \rightarrow \mathbb{R}$ by $$g(t) = \prod_{i: p_i^* > 0} \left(\frac{u_i(t)}{u_i^*} \right)^{p_i^*}$$ and 
\[ h(t) =  \prod_{i \in [n]} \left(\frac{p_i(t) \cdot B_i(t)^{\left(\frac{1-\alpha_i}{\alpha_i}\right)}}{\Bigl(\alpha_i \cdot p_i(t) + (1-\alpha_i) \cdot B_i(t)\Bigr)^{\frac{1}{\alpha_i}}}
\right)^{p_i^*}.\]
Then we have the following lemma.
\begin{lemma}
	\label{lem:iterative}
	With the lazy proportional response dynamics, for any vector of coefficients $\alpha_i$, where $\alpha_i \in [0, 1)$ for all $i$, we have the identity:
	\[ f(t+1) = f(t) \cdot g(t)\cdot h(t).\]
\end{lemma}
\begin{proof}	
	The bid update rule gives
	$$
	b_{i,j}(t+1) = \left( \frac{a_{i,j} \cdot x_{i,j}(t)}{u_i(t)} \right) \cdot B_i(t+1) = 
	\frac{a_{i,j}}{u_i(t)} \cdot \frac{b_{i,j}(t)}{p_j(t)} \cdot B_i(t+1)
	$$
	
	Expanding $z_{i,j}(t+1)$, we obtain
	\begin{align} \label{eq:z+ij_t+1}
	z_{i,j}(t+1) &= \left(\frac{b_{i,j}^*}{\left( \frac{a_{i,j} \cdot x_{i,j}(t)}{u_i(t)} \right) \cdot B_i(t+1)} \right)^{b_{i,j}^*} \notag \\
	& = \left(\frac{b_{i,j}^*}{b_{i,j}(t)} \cdot \frac{u_{i}(t)}{a_{i,j}} \cdot \frac{p_j(t)}{B_i(t+1)} \right)^{b_{i,j}^*} \notag \\
	& = z_{i,j}(t) \cdot \left( \frac{u_i(t)}{a_{i,j}} \cdot \frac{p_j(t)}{B_i(t+1)} \right)^{b_{i,j}^*}
	\end{align}
	Using the equilibrium property (\ref{eq:condition}) of the exchange economy and identity (\ref{eq:z+ij_t+1}) gives
	$$
	z_{i,j}(t+1) =
	\begin{cases}
	z_{i,j}(t) \cdot \left( \frac{u_i(t)}{u_i^*} \cdot \frac{p_j(t)}{p_j^*} \cdot \frac{p_i^*}{B_i(t+1)} \right)^{b_{i,j}^*} & if \; b_{i,j}^* > 0\\
	0 & otherwise
	\end{cases} 
	$$
	Expanding $w_i(t+1)$ gives 
	$$
	w_i(t+1) = \left(\frac{p_i^*}{B_i(t+1)} \right)^{p_i^* \cdot \left(\frac{1-\alpha_i}{\alpha_i}\right)} = \left(\frac{p_i^*}{B_i(t)} \right)^{p_i^* \cdot \left(\frac{1-\alpha_i}{\alpha_i}\right)} \cdot \left(\frac{B_i(t)}{B_i(t+1)} \right)^{p_i^* \cdot \left(\frac{1-\alpha_i}{\alpha_i}\right)}
	$$
	
	Expanding $f(t+1)$ yields
	\begin{small}
		\begin{align}
		f(t+1) & = \prod_{i,j \in [n]} z_{i,j}(t+1) \cdot \prod_{i \in [n]} w_{i}(t+1) \notag \\
		& = \prod_{i,j \in N: b_{i,j}^* > 0} z_{i,j}(t) \cdot \left( \frac{u_i(t)}{u_i^*} \cdot \frac{p_j(t)}{p_j^*} \cdot \frac{p_i^*}{B_i(t+1)} \right)^{b_{i,j}^*} \cdot  \prod_{i \in [n]} \left(\frac{p_i^*}{B_i(t)} \right)^{p_i^* \cdot \left(\frac{1-\alpha_i}{\alpha_i}\right)} \cdot \left(\frac{B_i(t)}{B_i(t+1)} \right)^{p_i^* \cdot \left(\frac{1-\alpha_i}{\alpha_i}\right)} \notag \\
		& =f(t) \cdot \left( \prod_{i,j \in N: b_{i,j}^* > 0} \left( \frac{u_i(t)}{u_i^*} \cdot \frac{p_j(t)}{p_j^*} \cdot \frac{p_i^*}{B_i(t+1)} \right)^{b_{i,j}^*} \right) \cdot \prod_{i \in [n]}   \left(\frac{B_i(t)}{B_i(t+1)} \right)^{p_i^* \cdot \left(\frac{1-\alpha_i}{\alpha_i}\right)} \notag
		\end{align}
	\end{small}
	Separating the utility terms from the price terms gives 
	\begin{small}
		\begin{align}
		f(t+1)	&= f(t) \cdot 
		\prod_{i: p_i^* > 0} \left(\frac{u_i(t)}{u_i^*} \right)^{\sum_{j} b_{i,j}^*}  
		\prod_{i: p_i^* > 0} \left( \frac{p_i^*}{B_i(t+1)}\right)^{\sum_{j} b_{i,j}^*}  
		\prod_{j: p_j^* > 0} \left( \frac{p_j(t)}{p_j^*}\right)^{\sum_{i} b_{i,j}^*}   \prod_{i \in [n]}   \left(\frac{B_i(t)}{B_i(t+1)} \right)^{p_i^* \cdot \left(\frac{1-\alpha_i}{\alpha_i}\right)}  \notag \\
		& = f(t) \cdot 
		\prod_{i: p_i^* > 0} \left(\frac{u_i(t)}{u_i^*} \right)^{p_i^*} \cdot 
		\prod_{i: p_i^* > 0} \left( \frac{p_i^*}{B_i(t+1)}\right)^{p_i^*} \cdot 
		\prod_{j: p_j^* > 0} \left( \frac{p_j(t)}{p_j^*}\right)^{p_j^*}  \cdot  \prod_{i \in [n]}   \left(\frac{B_i(t)}{B_i(t+1)} \right)^{p_i^* \cdot \left(\frac{1-\alpha_i}{\alpha_i}\right)} 
		\notag
		\end{align}
	\end{small}
	Note the condition on $i$ that there exists $j$ for which $b_{i,j}^* > 0$ is equivalent to $p_i^* > 0$, which at the market equilibrium is met for every index $i \in [n]$. Similarly, the condition on $j$ that there exists $i$ such that $b_{i,j}^* > 0$ is equivalent to $p_j^* > 0$. Then we can rewrite $f(t+1)$ as follows:
	\begin{align}
	f(t+1) & = f(t) \cdot 
	\prod_{i \in [n] } \left(\frac{u_i(t)}{u_i^*} \right)^{p_i^*} \cdot 
	\prod_{i \in [n]} \left( \frac{p_i^*}{B_i(t+1)}\right)^{p_i^*} \cdot 
	\prod_{j \in [n]} \left( \frac{p_j(t)}{p_j^*}\right)^{p_j^*}  \prod_{i \in [n]}   \left(\frac{B_i(t)}{B_i(t+1)} \right)^{p_i^* \cdot \left(\frac{1-\alpha_i}{\alpha_i}\right)}  \notag \\
	& = f(t) \cdot \prod_{i \in [n] } \left(\frac{u_i(t)}{u_i^*} \right)^{p_i^*} \cdot \prod_{i \in [n] } \left(\frac{p_i(t)}{B_i(t+1)} \right)^{p_i^*}   \prod_{i \in [n]}   \left(\frac{B_i(t)}{B_i(t+1)} \right)^{p_i^* \cdot \left(\frac{1-\alpha_i}{\alpha_i}\right)} 
	\end{align}
	
	We also analyze the following term in the identity for $f(t+1)$:
	\begin{align}
	& \prod_{i \in [n] } \left(\frac{p_i(t)}{B_i(t+1)} \right)^{p_i^*}  \cdot \; \prod_{i \in [n]}   \left(\frac{B_i(t)}{B_i(t+1)} \right)^{ \frac{p_i^* (1-\alpha_i)}{\alpha_i}}\notag \\
	& = \prod_{i \in [n]} \; \frac{p_i(t)^{p_i^*} \cdot B_i(t)^{\left(\frac{p_i^* (1-\alpha_i)}{\alpha_i}\right)}}{B_i(t+1)^{\frac{p_i^*}{\alpha_i}}} \notag \\
	& = \prod_{i \in [n]} \frac{p_i(t)^{p_i^*} \cdot B_i(t)^{\left(\frac{p_i^*(1-\alpha_i)}{\alpha_i}\right)}}{\Bigl(\alpha_i \cdot p_i(t) + (1-\alpha_i) \cdot B_i(t)\Bigr)^{\frac{p_i^*}{\alpha_i}}}\\
	& = h(t).
	\end{align}
	Recall that $g(t) = \prod_{i: p_i^* > 0} \left(\frac{u_i(t)}{u_i^*} \right)^{p_i^*}$. Then we get $f(t+1) = f(t) \cdot g(t) \cdot h(t)$ as required.
\end{proof}

\begin{lemma} \label{lem:f_lb}
	For any market instance, there is a constant $c > 0$, which is dependent on the market parameters, so that $f(t) \geq c$ for all $t$.
\end{lemma}
\begin{proof}
	For any pair $i,j \in [n]$, define a constant $\Delta_{i,j} > 0$ as follows: if $b_{i,j}^* = 0$, then $\Delta_{i,j} = 1$. Otherwise, since the money is normalized to sum up to $1$, we have $b_{i,j}(t) \leq 1$ for all $t$, so we can set 
	\begin{align}
	\Delta_{i,j} = \left( b_{i,j}^*\right)^{b_{i,j}^*}  \leq 
	\left( \frac{b_{i,j}^*}{b_{i,j}(t)}\right)^{b_{i,j}^*}   = z_{i,j}(t)
	\end{align}
	Then $z_{i,j}(t) \geq \Delta_{i,j}> 0$.
	Similarly, for any $i \in [n]$, if $p_i^* = 1$, then define a constant $\delta_i = 1$; otherwise, $\frac{p_i^*}{B_i(t)} \geq p_i^*$, so we can set 
	\begin{align}
	\delta_i = \left(p_i^*\right)^{\frac{p_i^* \cdot (1-\alpha_i)}{\alpha_i}} \leq 
	\left(\frac{p_i^*}{B_i(t)} \right)^{\frac{p_i^*\cdot(1-\alpha_i)}{\alpha_i}} = w_i(t) 
	\end{align}
	Let $c = \prod_{i,j \in [n]} \Delta_{i,j} \cdot \prod_{i \in [n]} \delta_i$.
	Then by definition of $f$, we have that $f(t) \geq c > 0$ for all $t$.
\end{proof}

We now show convergence of utilities. 
\begin{proof}[Proof of Theorem \ref{thm:util_me}]
	
	Observe that $\log{g(t)} = \sum_{i \in N} p_i^* \cdot \log{\left(\frac{u_i(t)}{u_i^*}\right)}$. Consider the Fisher market obtained by setting the budget of each player to the equilibrium price of its own good. Then the expression $\sum_{i \in N} p_i^* \cdot \log{u_i}$, which is the Eisenberg-Gale objective, is maximized by $u_i^*$. It follows that $\sum_{i \in N} p_i^* \cdot \log{u_i(t)} < \sum_{i \in N} p_i^* \cdot \log{u_i^*}$ whenever the utilities are different, i.e. $u_i(t) \neq u_i^*$ for some player $i$. Thus $g(t) \leq 1$ for all  $t$, and the equality holds if and only if $u_i(t) = u_i^*$ for all $i$. 
	
	Using the weighted arithmetic mean-geometric mean inequality, we have that $$p_i(t)^{\alpha_i} \cdot B_i(t)^{{1-\alpha_i}} \leq \alpha_i \cdot p_i(t) + (1-\alpha_i) \cdot B_i(t)$$ for all $\alpha_i \in (0, 1]$. Thus $h(t) \leq 1$ for all $t$.
	
	\medskip 
	
	Since $g(t) <1$ and $h(t) \leq 1$, we obtain that $f(t+1) < f(t) $ for all the times $t$ where the utilities are not the equilibrium utilities.
	It follows that $f(t)$ is monotonically decreasing. By Lemma~\ref{lem:f_lb}, there exists a constant $c$ so that $f(t) \geq c > 0$ for all $t$, so $\lim_{t \to \infty} f(t)$ exists and is non-zero. Taking the limit of $t \to \infty$ in the expression $f(t+1) = f(t) \cdot g(t) \cdot h(t)$, we get that $\lim_{t \to \infty} g(t) =1$, and so the dynamic reaches the market equilibrium utilities.
\end{proof}

\paragraph{Rate of convergence:} We show a bound on the rate of convergence in the following sense: 
as in the proof of Theorem \ref{thm:util_me}, 
we consider the Eisenberg-Gale objective function, with 
the budget of each player set to the equilibrium price of his own good. 
We then show that the average of the Eisenberg-Gale objective over 
$T$ rounds is within $O(1/T)$ of the optimum.
(This also implies that the Eisenberg-Gale objective at the 
average of all the allocations is close to optimum, 
due to concavity of the objective.)
In other words, the total difference from the optimum summed over all rounds 
is bounded by a constant.


\begin{corollary} \label{cor:convergence_rate} When each $\alpha_i \in [0, 1)$, the average of the Eisenberg-Gale objective over $T$ rounds converges to the optimum at a rate of $O(1/T)$.
\end{corollary}
\begin{proof}
	From Lemma \ref{lem:iterative}, we have that $\log(f(t+1)) = \log(f(t)) + \log(g(t)) + \log(h(t))$. Recall from the proof of Theorem \ref{thm:util_me} that $\log(h(t)) \leq 0$ for all $t$. Then we get 
	\begin{align} 
	& \log(f(t+1)) \leq \log(f(t)) + \log(g(t)) \implies \notag \\
	& - \log(g(t)) \leq \log(f(t)) - \log(f(t+1)) \notag
	\end{align}
	The difference in the Eisenberg-Gale objective between the optimum and the current iteration is exactly $-\log(g(t))$. Summing this inequality over all rounds from $1$ to $T$, we get 
	\begin{align} \label{eq:log_g_sum}
	\sum_{t=1}^T - \log(g(t)) \leq \log(f(1)) - \log(f(T + 1))
	\end{align}
	We have that $\log(f(t)) \geq 0$ for all $t$. Then from (\ref{eq:log_g_sum}), we obtain 
	$$
	\sum_{t=1}^T \frac{- \log(g(t))}{T} \leq \frac{\log(f(1))}{T}
	$$
	Since $\log(f(1))$ is a constant, we get that $\sum_{t=1}^T -\log(g(t))/T \in O(1/T)$ as required.
\end{proof}

\subsection{Convergence of Bids and Allocations in the Lazy Dynamic}

In this subsection, we focus on the case when $\alpha_i < 1$ for all $i$. 
In this case, the limit points of the sequence must correspond to equilibrium prices. 

\begin{theorem}
	\label{thm:limitprices}
	If each $\alpha_i < 1$, then any limit point of the sequence $\bij(t)$  is an equilibrium. 
\end{theorem} 
\begin{proof}
	Let $\vec{b}$ be a limit point of the bids with a converging subsequence $s = (s_1, s_2, \ldots)$. 
	As before, the utilities when the players use these bids converge to the equilibrium utility profile. 
	Further, we have that $lim_{n\rightarrow \infty} h(s_n)  = 1$, 
	which implies $lim_{n\rightarrow \infty} B_i(s_n)  = lim_{n\rightarrow \infty} p_i(s_n)$ for all $i$. 
	Thus for the limit point we must have that $\forall~i, \sum_j \bij = p_i$. 
	It is now easy to verify that the allocations and prices induced by the bids satisfy the equilibrium conditions.  
\end{proof}

With this in hand, we can now show that the equilibrium bids and allocations converge. Note that this automatically implies that the prices converge too. 
\begin{theorem}
	\label{thm:lazy_convergence} 
	If each $\alpha_i < 1$, then the bids and allocations of the lazy dynamics converge to a market equilibrium.
\end{theorem}
\begin{proof}
	Suppose towards a contradiction that there exists a starting configuration $b_{i,j}(0)$ for which the sequence of bids does not converge.
	Then since the set of feasible bid matrices is compact, there exist two subsequences of rounds $s' = (s'_1, s'_2, \ldots)$ and $s'' = (s''_1, s''_2, \ldots)$ 
	so that the bids converge to two limit bid matrices $\vec{b}'$ and $\vec{b}''$ along these subsequences, respectively.
	
	From Theorem~\ref{thm:limitprices} we have that both $\vec{b}'$ and $\vec{b}''$ are market equilibrium bids. The main idea is to consider these two limit points. Both of these  will satisfy the equilibrium allocation conditions. However, the two sets of bids also have to be different, which will give a contradiction.
	
	Define the function $f$ with respect to the limit bids $\vec{b}'$ (i.e. set $\vec{b}^* = \vec{b}'$ in the definition of $f$). Then we have that $\log{f(\vec{b}')} = 0$ and $\log{f(\vec{b})} > 0$ for all $\vec{b}\neq \vec{b}'$. 
	Let $\vec{p}'$ and $\vec{p}''$ be the price vectors corresponding to bids $\vec{b}'$ and $\vec{b}''$. Since the bids along the sequences $s'$ and $s''$ converge to different limits, by continuity of the KL divergence there exist $\delta > 0$ and an index $k_0 \in \mathbb{N}$ so that $ |\log{f(\vec{b}(s'_k))} - \log{f(\vec{b}(s''_k))}| > \delta$ for all $k \geq k_0$. This implies that $\lim_{k \to \infty} \log{f(\vec{b}(s''_k))} >0$, which is a contradiction. Thus the subsequence $s''$ cannot exist and the bids converge in the limit as required.
\end{proof}

\section{Cycling behavior and convergence of allocations in the non-lazy dynamics}
\label{sec:cycling} 

Note the bids may in fact cycle in the proportional response dynamics (the strictly non-lazy version), thus we do not necessarily obtain in the limit the market equilibrium prices. Consider the following economy.
\begin{example}[Bid cycling in the non-lazy dynamic]
	Let $N = \{1, 2\}$, with utilities $a_{1,1} = a_{2,2} = 0$, $a_{1,2} = a_{2,1} = 1$ and initial bids $b_{1,1}(0) = b_{2,2}(0) = 0$, $b_{1,2}(0) = 1/3$, $b_{2,1}(0) = 2/3$. Then at the market equilibrium the prices of the two goods are equal, so $p_1^* = p_2^* = 1/2$, with $b_{1,2}^* = b_{2,1}^* = 1/2$. Thus we have $f(0) = \left( \frac{1/2}{1/3} \right)^{1/2} \cdot \left( \frac{1/2}{2/3}\right)^{1/2}$. Since the players swap their budgets throughout the dynamic, we get that $f(2k+1) = \left( \frac{1/2}{2/3} \right)^{1/2} \cdot \left( \frac{1/2}{1/3}\right)^{1/2} = f(2k) = f(0)$ for all $k \in \mathbb{N}$.
\end{example}

More generally, prices do not converge in bipartite graphs when the two sides are unbalanced. For example, consider any economy where the underlying graph is bipartite. Suppose the initial sums of the budgets on the two sides of the graph are different. Then the prices do not converge. This follows from the fact that the two sides will keep swapping their money in each iteration. 

This phenomenon is analogous to what happens with periodic Markov chains.
We note that cycling can happen even if the valuation matrix has all the entries non-zero and the bids are strictly positive. Rather, what determines cycling are the consumption graph in the equilibrium allocation together with the initial distribution of bids. 

\medskip

We first show that the allocation always converges and 
later characterize the instances where the bids cycle, and start by showing a basic property of the market equilibrium (for which we could not find a reference).
The following proposition states that for a linear exchange market, any equilibrium allocation can be paired with any equilibrium price to get an equilibrium pair of allocation and price. 
\begin{proposition}
	\label{thm:universal_allocation} 
	Let $\vec{p^*}$ be any equilibrium price, and $\vec{u^*}$ be equilibrium utilities. 
	Let $\vec{x}$ be any feasible allocation that gives equilibrium utilities to all players, i.e., 
	$\forall~i$, $\sum_j \aij \xij = u_i^*$. 
	Then the pair $(\vec{x}, \vec{p^*})$ is an equilibrium. 
\end{proposition}
\begin{proof}
	Let $(\vec{x^*}, \vec{p^*})$ be a pair of equilibrium allocation and price for the exchange market. 
	Consider the Fisher market with budgets $B_i = p_i^*$ for all $i$, denoted by $\fisher(\vec{p^*})$. 
	\begin{enumerate}
		\item Then the pair $(\vec{x^*}, \vec{p^*})$ is also an equilibrium of this Fisher market, since the equilibrium conditions for the exchange market 
		directly imply the equilibrium conditions for the Fisher market. 
		\item Since Fisher markets have a unique equilibrium price \cite{eisenberg1959consensus}, this price must be $\vec{p^*}$. 
		\item Now suppose $\vec{x}$  be any other allocation as in the hypothesis of the Theorem. This implies that $\vec{x}$ is an optimal solution to the Eisenberg-Gale convex program corresponding to $\fisher(\vec{p^*} )$, and therefore $(\vec{x}, \vec{p^*})$ is also an equilibrium of $\fisher(\vec{p^*} )$. 
		\item This implies that  $(\vec{x}, \vec{p^*})$ is also an equilibrium for the exchange market. Once again, the equilibrium conditions are essentially identical. 
	\end{enumerate}
\end{proof}

\begin{proposition}
	\label{thm:unique_prices} 
	Suppose that there is an equilibrium allocation such that the support graph on the set of nodes is connected. Then the equilibrium prices are unique up to a scaling factor
\end{proposition}
\begin{proof}
	Suppose that $i$ and $j$ are such that $\xij > 0$. 
	Then from the condition (\ref{eq:condition}), we get that the ratio of their equilibrium prices must equal $\aij/u_i^*$, which is independent of the choice of the equilibrium, since equilibrium utilities are unique. 
	Thus if the support graph is connected, the ratio of any two equilibrium prices remains the same, which means the equilibrium prices form a ray. 
\end{proof}

\begin{theorem}
	\label{thm:nonlazy_allocation}
	The allocation in the (non-lazy) proportional response dynamics converges to a market equilibrium allocation.
\end{theorem}
\begin{proof}
	For any initial bids $b_{i,j}(0)$, there is a subsequence of bids converging to some limit $\vec{b}'$. We note that the limit $\vec{b}'$ may not be market equilibrium bids. 
	However, the allocation corresponding to $\vec{b}'$ must give equilibrium utilities, by Theorem~\ref{thm:util_me}. 
	Any such allocation is also an equilibrium, see Proposition~\ref{thm:universal_allocation}.
	Let this allocation be $\vec{x}'= \vec{x^*}$, the corresponding equilibrium price be $\vec{p^*}$ and the corresponding bids be $\vec{b^*}$. 

	Let $\vec{p}'$ be the prices induced by the bids $\vec{b}'$. Since the allocation under the two bid profiles is the same, we get $$
	\frac{b_{i,j}'}{b_{i,j}^*} = \frac{\left(\frac{b_{i,j}'}{p_j'}\right)}{\left(\frac{b_{i,j}^*}{p_j^*}\right)} \cdot \frac{p_j'}{p_j^*} = \frac{x_{i,j}'}{x_{i,j}^*} \cdot \frac{p_j'}{p_j^*} = \frac{p_j'}{p_j^*}
	$$
	
	Let $j \in [n]$ be arbitrary but fixed. Then we obtain 
	\begin{align} \label{eq:bprime_bstar}
	\sum_{i \in [n]} b_{i,j}^* \cdot \log{\frac{b_{i,j}'}{b_{i,j}^*}} = \sum_{i \in [n]} b_{i,j}^* \cdot \log{\frac{p_j'}{p_j^*}} = p_j^* \cdot \log{\frac{p_j'}{p_j^*}} 
	\end{align}
	Suppose for the sake of contradiction that the sequence of allocations of the dynamic does not converge. Then there exists another subsequence of bids converging to a different limit $\vec{b}'' \neq \vec{b}'$ which has the property that $\vec{x}'' \neq \vec{x}'$, where $\vec{x}''$ is the allocation at the bid profile $\vec{b}''$. 
	We use an identity which we state in Lemma~\ref{lem:kldecomposition} to obtain
	\[
	\sum_{i \in [n]} b_{i,j}^* \cdot \log{\frac{b_{i,j}''}{b_{i,j}^*}} =
	p_j^* \cdot \log{\frac{p_j''}{p_j^*}} + p_j^* \cdot \sum_{i \in [n]} x_{i,j}' \cdot \log{\frac{x_{i,j}''}{x_{i,j}^*}}
	.\]
	
	From the proof of Theorem \ref{thm:util_me}, the KL divergence between the fixed point bids $\vec{b^*}$ and the dynamic bids $\vec{b}(t)$ is decreasing and converges to some constant $\alpha \geq 0$. Since both $\vec{b}', \vec{b}''$ are limit points of $\vec{b}(t)$, this implies that 
	\begin{align} 
	&\sum_{j \in [n]} \sum_{i \in [n]} b_{i,j}^* \cdot \log{\frac{b_{i,j}'}{b_{i,j}^*}} = \sum_{j \in [n]} \sum_{i \in [n]} b_{i,j}^* \cdot \log{\frac{b_{i,j}''}{b_{i,j}^*}} \iff \notag \\
	& \sum_{j \in [n]} p_j^* \cdot \log{\frac{p_j'}{p_j^*}}  = \sum_{j \in [n]} p_j^* \cdot \log{\frac{p_j''}{p_j^*}} +\sum_{j \in [n]} p_j^* \cdot \sum_{i \in [n]} x_{i,j}' \cdot \log{\frac{x_{i,j}''}{x_{i,j}'}}
	\end{align}
	Among all possible choices of limit bid profiles $\vec{b}'$ and market equilibrium bids $\vec{b^*}$ with the same allocation, select the pair $(\vec{b}', \vec{b^*})$ that minimizes the sum $\sum_{j \in [n]} p_j^* \cdot \log{\frac{p_j'}{p_j^*}}$; this is possible since the bid space is compact so the infimum of a set of accummulation points is itself an accummulation point (See Lemma 4.1 in \cite{khalil2002nonlinear} for example). 
	For this choice of bids, we obtain  
	\begin{align}
	\sum_{j \in [n]} p_j^* \cdot \log{\frac{p_j''}{p_j^*}} \geq \sum_{j \in [n]} p_j^* \cdot \log{\frac{p_j'}{p_j^*}} =\sum_{j \in [n]} p_j^* \cdot \log{\frac{p_j''}{p_j^*}} +  \sum_{j \in [n]} p_j^* \cdot \sum_{i \in [n]} x_{i,j}' \cdot \log{\frac{x_{i,j}''}{x_{i,j}'}}
	\end{align}
	Note that each term $\sum_{i \in [n]} x_{i,j}' \cdot \log{\frac{x_{i,j}''}{x_{i,j'}}}$ represents the KL-divergence between the allocation of good $i$ at the limits $x'$ and $x''$. Since the KL divergence is always non-negative, we get that in fact $\sum_{i \in [n]} x_{i,j}' \cdot \log{\frac{x_{i,j}''}{x_{i,j}'}} = 0$ for each $j$, so the allocation at the limit $\vec{b}''$ is the same as at $\vec{b}'$. Thus the assumption that the limit $\vec{b}''$ had a different allocation from $\vec{b}'$ was incorrect, so the allocations must converge.
\end{proof}

\begin{lemma}
	\label{lem:kldecomposition} 
	For any two bids $\vec{b}$ and $\vec{b'}$ and corresponding allocations and prices, we have that 
	\[
	\sum_{i \in [n]} b_{i,j} \cdot \log{\frac{b_{i,j}'}{b_{i,j}}} =
	p_j \cdot \log{\frac{p_j'}{p_j}} + p_j \cdot \sum_{i \in [n]} x_{i,j} \cdot \log{\frac{x_{i,j}'}{x_{i,j}}}
	.\]
\end{lemma}
\begin{proof}
We start with the left hand side of the identity and rewrite it as follows:
	\begin{align} \label{eq:diff_limit_bids}
	\sum_{i \in [n]} b_{i,j} \cdot \log{\frac{b_{i,j}'}{b_{i,j}}} &= \sum_{i \in [n]} b_{i,j} \cdot  \log{ \frac{x_{i,j}' \cdot p_j'}{x_{i,j} \cdot p_j} } \notag \\
	& = \sum_{i \in [n]} b_{i,j} \cdot \log{\frac{p_j'}{p_j}} + \sum_{i \in [n]} b_{i,j} \cdot \log{ \frac{x_{i,j}'}{x_{i,j}}} \notag \\
	& = p_j \cdot \log{\frac{p_j'}{p_j}} + \sum_{i \in [n]} p_j \cdot x_{i,j} \cdot \log{\frac{x_{i,j}'}{x_{i,j}}} \notag \\
	& = p_j \cdot \log{\frac{p_j'}{p_j}} + p_j \cdot \sum_{i \in [n]} x_{i,j} \cdot \log{\frac{x_{i,j}'}{x_{i,j}}}
	\end{align}
\end{proof}


We now characterize the limit cycles in the price space. 
Note that we already have from Theorem \ref{thm:nonlazy_allocation} 
that the allocation must remain an invariant along any limit cycle. 
\begin{theorem}\label{thm: limit alloc_cycling_bids}
	The limit bids of the proportional response dynamics are either an equilibrium or there exist equivalence classes $C_1, \ldots, C_k$, where $C_i \cap C_j = \emptyset$, $C_i \subseteq N$ for all $i$, such that there exists $\lambda_i(t) \geq 0$ for each $C_i$ with the property that the price of each good $j \in C_i$ satisfies 
	$p_j(t)/p_j^* = \lambda_i(t), $ for some 
	equilibrium price $\vec{p^*}$. 
\end{theorem}
\begin{proof}
	We have from Theorem \ref{thm:nonlazy_allocation} that the allocation converges to an equilibrium; 
	let $\vec{p^*}$ be the corresponding equilibrium price, 
	and $\vec{u^*}$ the equilibrium utilities.
	Consider a limit point, and consider taking one step of the dynamics from the limit point. We denote the limit point by $\vec{b}(t)$, 
	just to indicate that the next step after this is $\vec{b}(t+1)$.  
	The point $\vec{b}(t+1)$ is also a limit point of the sequence.\footnote{This is a well known property of continuous time dynamical systems, and the same holds for discrete time systems as well. A quick proof: if $s_n$ is the subsequence whose limit is $\vec{b}(t)$, then by continuity of the update rule, we get that the subsequence  $s_{n}+1$ must have $\vec{b}(t+1)$ as its limit.} 
	By the definition of the update rule, we have that
	$$
	u_i(t+1) = \sum_{j = 1}^n a_{i,j} \cdot x_{i,j}(t+1) = \sum_{j=1}^n a_{i,j} \cdot \frac{b_{i,j}(t+1)}{\sum_{k=1}^n b_{k,j}(t+1)} 
	$$
	By Theorem \ref{thm:util_me}, the utilities converge in the limit to the market equilibrium utilities,  therefore we have that $u_i(t) = u_i^* = u_i(t+1)$ for all $ i \in [n]$. Then we get
	\begin{align} \label{eq:fixedu}
	u_i(t+1) & =\sum_{j=1}^n a_{i,j} \cdot \frac{b_{i,j}(t+1)}{\sum_{k=1}^n b_{k,j}(t+1)}  \notag \\
	& = \sum_{j=1}^n a_{i,j} \cdot \frac{\frac{a_{i,j}}{u_i^*} \cdot x_{i,j}(t) \cdot B_i(t+1)}{\sum_{k=1}^n \frac{a_{k,j}}{u_k^*} \cdot x_{k,j}(t) \cdot B_k(t+1)} \notag \\
	& = \sum_{j=1}^n a_{i,j} \cdot \frac{\frac{p_j^*}{p_i^*} \cdot x_{i,j}(t) \cdot B_i(t+1)}{\sum_{k=1}^n \frac{p_j^*}{p_k^*} \cdot x_{k,j}(t) \cdot B_k(t+1)} \notag \\
	& = \sum_{j=1}^n \; \frac{\frac{a_{i,j}}{p_i^*} \cdot x_{i,j}(t) \cdot B_i(t+1)}{\sum_{k=1}^n x_{k,j}(t) \cdot \frac{B_k(t+1)}{p_k^*}} \notag \\
	& = \sum_{j=1}^n \frac{a_{i,j} \cdot x_{i,j}(t) \cdot \left(\frac{p_i(t)}{p_i^*} \right)}{\sum_{k=1}^n x_{k,j}(t) \cdot \left( \frac{p_k(t)}{p_k^*} \right)} = u_i(t) = \sum_{j=1}^n a_{i,j} \cdot x_{i,j}(t)
	\end{align}
	
	The third equality follows from (\ref{eq:condition}). 
	Let $\lambda_i(t) = p_i(t) / p_i^*$. Then identity (\ref{eq:fixedu}) is equivalent to the following system of equations, where $a_{i,j}, x_{i,j}(t)$ are given and $\lambda_i(t)$ are variables:
	\[
	\left\{
	\begin{array}{ll}
	\lambda_i(t) \geq 0, \mbox{for all} \; i \in [n] \\
	\sum_{j=1}^n 
	a_{i,j} \cdot x_{i,j}(t) \cdot \left( \frac{\lambda_i(t)}{\sum_{k=1}^n x_{k,j}(t) \cdot \lambda_k(t)} - 1 \right) = 0, \; \; \mbox{for all} \; i \in [n] 
	\end{array}
	\right.
	\]
	One solution can be obtained as follows. Define the equivalence relation $\sim$ as follows: $i \sim k$ if and only if there exists $ j \in [n]$ such that $x_{i,j}, x_{k,j} > 0$. Then consider the transitive closure of this graph -- that is, if player $i$ purchases some other good $\ell \neq j$, then all the players that purchase strictly positive amounts of good $\ell$ are in the same equivalence class with $i$ and $k$. 
	Let $C_1, \ldots, C_k$ be equivalence classes with respect to the $\sim$ relation. Then setting $\lambda_i(t) = \lambda_k(t)$ for each $i \sim k$ works. This means that all the goods in the same equivalence class have prices within the same factor away from the market equilibrium price at any point in time. 
	
	We show that in fact these are the only solutions. Consider an arbitrary solution to this system and suppose towards a contradiction that there exist two players $i,i'$ in the same equivalence class $C_{\ell}$ but with $\lambda_i(t) > \lambda_{i'}(t)$. W.l.o.g., $\lambda_i(t) = \max_{v \in C_{\ell}} \lambda_v(t)$. Then for all $j$ with $x_{i,j}(t) > 0$ we have 
	\begin{align}  \label{eq:contradiction_lambda}
	\frac{\lambda_{i}(t)}{\sum_{k=1}^n x_{k,j}(t) \cdot \lambda_k(t)} > 1 \iff a_{i,j} \cdot x_{i,j}(t) \cdot \left(\frac{\lambda_{i}(t)}{\sum_{k=1}^n x_{k,j}(t) \cdot \lambda_k(t)} - 1\right) > 0
	\end{align}
	Summing up inequality (\ref{eq:contradiction_lambda}) over all $j$ we get  
	$$
	\sum_{j =1}^n a_{i,j} \cdot x_{i,j}(t) \cdot \left(\frac{\lambda_{i}(t)}{\sum_{k=1}^n x_{k,j}(t) \cdot \lambda_k(t)} - 1\right) > 0,
	$$
	which does not satisfy the required system of identities. Thus the assumption must have been false and $\lambda_i(t) = \lambda_{i'}(t)$ for any players $i,i'$ in the same equivalence class.
\end{proof}

\section{Comparison between the Wu-Zhang dynamics and Proportional Response Dynamics} \label{apx:comparison}

In this section we compare the proportional response dynamics we studied with the Wu-Zhang dynamics~\cite{wu2007proportional}. Given that the dynamics in Wu-Zhang does not use money and the players play by matching the offers from the other players in previous rounds, we refer to the Wu-Zhang dynamics as tit-for-tat.

The two dynamics may have different trajectories given the same starting configuration, even in the special case where there exists a vector $\vec{w} = (w_1, \ldots, w_n)$ so that $a_{i,j} = w_j$ for each player $i$ and each good $j$, which is the special case where Wu and Zhang~\cite{wu2007proportional} established convergence to market equilibria for the tit-for-tat dynamic.

For valuation matrix $A = (a_{i,j})$, the tit-for-tat dynamic (defined in Wu-Zhang~\cite{wu2007proportional}) is 
$$ \textstyle
y_{i,j}(t+1) = \frac{y_{j,i}(t) \cdot a_{i,j}}{u_i(t)}
$$
where $y_{j,i}(t)$ is the fraction received by player $i$ from good $j$ in round $t$ and the utility of player $i$ in round $t$ is $u_i(t) = \sum_{k=1}^n y_{k,i}(t) \cdot a_{i,k}$.
(The order of the subscripts here is good, player, which is the convention used in \cite{wu2007proportional}, as opposed to our notation where the order is player, good.)

\begin{figure}[H]
	\centering
	\subfigure[Allocations under tit-for-tat.]
	{
		\includegraphics[scale=0.33]{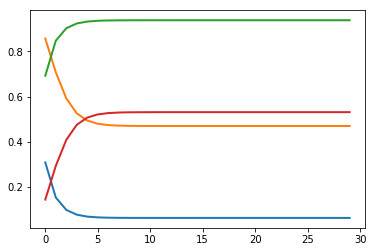}
		\label{fig:tit-for-tat-comparison}
	}
	\subfigure[Allocations under the non-lazy proportional dynamic.]
	{
		\includegraphics[scale = 0.33]{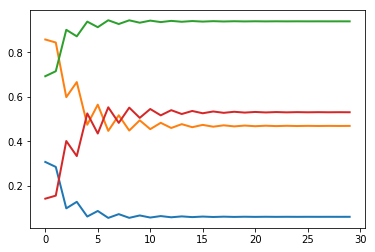}
		\label{fig:PR-comparison}
	}
	\subfigure[Allocations under the lazy proportional response dynamic with all $\alpha_i = 1/2$.]
	{
		\includegraphics[scale = 0.33]{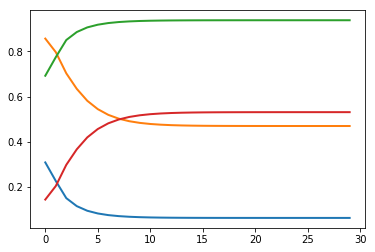}
		\label{fig:PR-comparison-lazy}
	}
	\subfigure[Utilities under tit-for-tat.]
	{
		\includegraphics[scale=0.33]{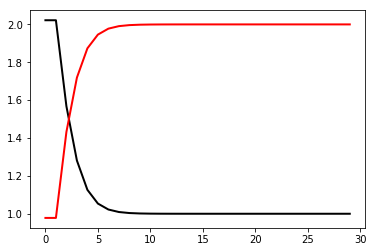}
		\label{fig:tit-for-tat-comparison_utilities}
	}
	\subfigure[Utilities under the non-lazy proportional dynamic.]
	{
		\includegraphics[scale = 0.33]{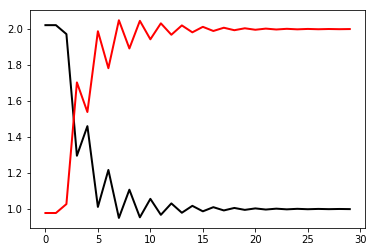}
		\label{fig:PR-comparison_utilities}
	}
	\subfigure[Utilities under the lazy proportional response dynamics with all $\alpha_i = 1/2$.]
	{
		\includegraphics[scale = 0.33]{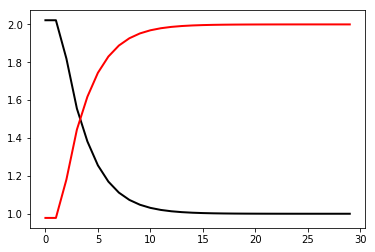}
		\label{fig:PR-comparison-lazy_utilities}
	}
	\caption{Comparison between tit-for-tat, the non-lazy proportional dynamic, and the lazy proportional dynamic, initialized with the same initial fractions of splitting the goods. The valuation matrix is $\vec{a} = [[1, 2], [1, 2]]$ and the initial bids for the proportional response dynamics executions are $\vec{b} =  [[0.4, 0.6], [0.9, 0.1]]$. The initial fractions for tit-for-tat are given by $y_{j,i}(0) = b_{i,j}(0)/p_j(0)$, where $y_{j,i}(0)$ is the fraction of good $j$ that player $i$ receives in round $0$ for the bids $\vec{b}$.}
	\label{fig:comparison}
\end{figure}

In Figure~\ref{fig:comparison}, the $X$ axis shows the round number, while the $Y$ axis shows each utility $u_i(t)$ over time and each allocation over time, where an allocation means the fraction received by each player $i$ from good $j$, for each $i,j$, for every time unit $t = 0,1,2, \infty$.

\bigskip
In Figure~\ref{fig:tit-for-tat-cyling}, we show a three player economy on which the tit-for-tat dynamic cycles. To obtain the cycling, we consider valuations 
\[\textbf{a} = 
\begin{vmatrix}
0 & 6 & 4 \\ 
3 & 0 & 9 \\        
9 & 6 & 0\\
\end{vmatrix}
\]
and initial fractions 
\[\textbf{y}(0) = 
\begin{vmatrix}
0.0 & 0.2805339037254016 & 0.7194660962745985 \\
0.273923422472049 & 0.0 & 0.726076577527951 \\
0.491752727261851 & 0.5082472727381491 & 0.0\\
\end{vmatrix}
\]

\begin{figure}[H]
	\centering
	\subfigure[Allocations under tit-for-tat for $T=30$ rounds.]
	{
		\includegraphics[scale=0.45]{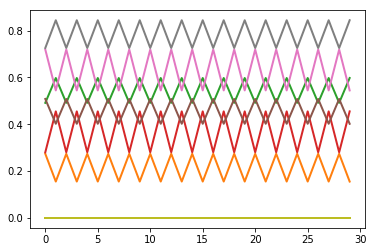}
		\label{fig:tit-for-tat-allocations-cycling-100}
	}
	\subfigure[Utilities under tit-for-tat for $T=30$ rounds; same initial conditions as in Figure (a)]
	{
		\includegraphics[scale=0.45]{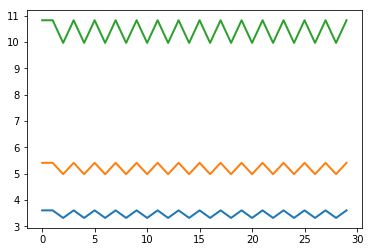}
		\label{fig:tit-for-tat-utilities-cycling-100}
	}
	\caption{Tit-for-tat dynamic cycling for both utilities and allocations with period two. The instance is a three player economy where the matrix does not satisfy the symmetry property under which tit-for-tat is known to converge to market equilibria~\cite{wu2007proportional} (Recall this property requires that each good has a common value, i.e. that there exist values $w_j$ so that $a_{i,j} = w_j$ for each player $i$ and good $j$.).}
	\label{fig:tit-for-tat-cyling}
\end{figure}

The fractions oscillate between the values at time $t=0$, equal to $\textbf{y}(0)$, and those from time $t=1$, which are equal to:

\[\textbf{y}(1) = 
\begin{vmatrix}
0.0 & 0.4552048517736218 & 0.5447951482263783 \\
0.1553967077250424 & 0.0 & 0.8446032922749576 \\
0.5978029457196989 & 0.402197054280301 & 0.0\\
\end{vmatrix}
\]

We note that market equilibria exist on such graphs and proportional response dynamics converges on these instances for any strictly positive initial bids.

\section{Discussion}

It would be interesting to see whether a generalization of the dynamic converges for the exchange economy when each player brings multiple goods, as well as to understand more broadly what families of dynamics lead to efficient exchange. It would also be interesting to show a $1/T$ rate of convergence of any last iterate.

 \section{Acknowledgements}
 
 We thank the reviewers for useful feedback that helped improve the paper.

\bibliography{exchangebib}


\begin{thebibliography}{68}


\ifx \showCODEN    \undefined \def \showCODEN     #1{\unskip}     \fi
\ifx \showDOI      \undefined \def \showDOI       #1{#1}\fi
\ifx \showISBNx    \undefined \def \showISBNx     #1{\unskip}     \fi
\ifx \showISBNxiii \undefined \def \showISBNxiii  #1{\unskip}     \fi
\ifx \showISSN     \undefined \def \showISSN      #1{\unskip}     \fi
\ifx \showLCCN     \undefined \def \showLCCN      #1{\unskip}     \fi
\ifx \shownote     \undefined \def \shownote      #1{#1}          \fi
\ifx \showarticletitle \undefined \def \showarticletitle #1{#1}   \fi
\ifx \showURL      \undefined \def \showURL       {\relax}        \fi
\providecommand\bibfield[2]{#2}
\providecommand\bibinfo[2]{#2}
\providecommand\natexlab[1]{#1}
\providecommand\showeprint[2][]{arXiv:#2}

\bibitem[\protect\citeauthoryear{Arrow, Block, and Hurwicz}{Arrow
  et~al\mbox{.}}{1959}]%
        {arrow1959stability}
\bibfield{author}{\bibinfo{person}{K.~J. Arrow}, \bibinfo{person}{H.~D. Block},
  {and} \bibinfo{person}{L. Hurwicz}.} \bibinfo{year}{1959}\natexlab{}.
\newblock \showarticletitle{On the stability of the competitive equilibrium,
  {II}}.
\newblock \bibinfo{journal}{\emph{Econometrica: Journal of the Econometric
  Society}} (\bibinfo{year}{1959}), \bibinfo{pages}{82--109}.
\newblock


\bibitem[\protect\citeauthoryear{Ausubel}{Ausubel}{2004}]%
        {Ausubel04}
\bibfield{author}{\bibinfo{person}{L.~M. Ausubel}.}
  \bibinfo{year}{2004}\natexlab{}.
\newblock \showarticletitle{An efficient ascending-bid auction for multiple
  objects}.
\newblock \bibinfo{journal}{\emph{Am. Econ. Rev}} \bibinfo{volume}{94},
  \bibinfo{number}{5} (\bibinfo{year}{2004}), \bibinfo{pages}{1452--1475}.
\newblock


\bibitem[\protect\citeauthoryear{Avigdor-Elgrabli, Rabani, and
  Yadgar}{Avigdor-Elgrabli et~al\mbox{.}}{2014}]%
        {avigdor2014convergence}
\bibfield{author}{\bibinfo{person}{N. Avigdor-Elgrabli}, \bibinfo{person}{Y.
  Rabani}, {and} \bibinfo{person}{G. Yadgar}.} \bibinfo{year}{2014}\natexlab{}.
\newblock \showarticletitle{Convergence of T\^atonnement in {F}isher Markets}.
\newblock \bibinfo{journal}{\emph{arXiv preprint arXiv:1401.6637}}
  (\bibinfo{year}{2014}).
\newblock


\bibitem[\protect\citeauthoryear{Babaioff, Blumrosen, and Nisan}{Babaioff
  et~al\mbox{.}}{2017}]%
        {BBN17}
\bibfield{author}{\bibinfo{person}{M. Babaioff}, \bibinfo{person}{L.
  Blumrosen}, {and} \bibinfo{person}{N. Nisan}.}
  \bibinfo{year}{2017}\natexlab{}.
\newblock \showarticletitle{Selling Complementary Goods: Dynamics, Efficiency
  and Revenue}. In \bibinfo{booktitle}{\emph{ICALP}}.
  \bibinfo{pages}{134:1--134:14}.
\newblock


\bibitem[\protect\citeauthoryear{Bei, Garg, and Hoefer}{Bei
  et~al\mbox{.}}{2019}]%
        {bei2019ascending}
\bibfield{author}{\bibinfo{person}{X. Bei}, \bibinfo{person}{J. Garg}, {and}
  \bibinfo{person}{M. Hoefer}.} \bibinfo{year}{2019}\natexlab{}.
\newblock \showarticletitle{Ascending-price algorithms for unknown markets}.
\newblock \bibinfo{journal}{\emph{ACM Trans. Alg.}} \bibinfo{volume}{15},
  \bibinfo{number}{3} (\bibinfo{year}{2019}), \bibinfo{pages}{37}.
\newblock


\bibitem[\protect\citeauthoryear{Bhawalkar and Roughgarden}{Bhawalkar and
  Roughgarden}{2011}]%
        {bhawalkar2011welfare}
\bibfield{author}{\bibinfo{person}{K. Bhawalkar} {and} \bibinfo{person}{T.
  Roughgarden}.} \bibinfo{year}{2011}\natexlab{}.
\newblock \showarticletitle{Welfare guarantees for combinatorial auctions with
  item bidding}. In \bibinfo{booktitle}{\emph{SODA}}.
  \bibinfo{pages}{700--709}.
\newblock


\bibitem[\protect\citeauthoryear{Birnbaum, Devanur, and Xiao}{Birnbaum
  et~al\mbox{.}}{2011}]%
        {BDX11}
\bibfield{author}{\bibinfo{person}{B. Birnbaum}, \bibinfo{person}{N.~R.
  Devanur}, {and} \bibinfo{person}{L. Xiao}.} \bibinfo{year}{2011}\natexlab{}.
\newblock \showarticletitle{Distributed Algorithms via Gradient Descent for
  {F}isher Markets}. In \bibinfo{booktitle}{\emph{Proc. of the 12th ACM Conf.
  on Electronic Commerce}}. \bibinfo{pages}{127--136}.
\newblock


\bibitem[\protect\citeauthoryear{Brainard and Scarf}{Brainard and
  Scarf}{2005}]%
        {BS05}
\bibfield{author}{\bibinfo{person}{W.~C. Brainard} {and} \bibinfo{person}{H.~E.
  Scarf}.} \bibinfo{year}{2005}\natexlab{}.
\newblock \showarticletitle{How to Compute Equilibrium Prices in 1891}.
\newblock \bibinfo{journal}{\emph{The American Journal of Economics and
  Sociology}} \bibinfo{volume}{64}, \bibinfo{number}{1} (\bibinfo{year}{2005}),
  \bibinfo{pages}{57--83}.
\newblock


\bibitem[\protect\citeauthoryear{Branzei and Filos-Ratsikas}{Branzei and
  Filos-Ratsikas}{2019}]%
        {BFR19}
\bibfield{author}{\bibinfo{person}{Simina Branzei} {and} \bibinfo{person}{Aris
  Filos-Ratsikas}.} \bibinfo{year}{2019}\natexlab{}.
\newblock \showarticletitle{Walrasian Dynamics in Multi-unit Markets}. In
  \bibinfo{booktitle}{\emph{AAAI}}.
\newblock


\bibitem[\protect\citeauthoryear{Branzei, Mehta, and Nisan}{Branzei
  et~al\mbox{.}}{2018}]%
        {BMN18}
\bibfield{author}{\bibinfo{person}{S. Branzei}, \bibinfo{person}{R. Mehta},
  {and} \bibinfo{person}{N. Nisan}.} \bibinfo{year}{2018}\natexlab{}.
\newblock \showarticletitle{Universal Growth in Production Economies}.
\newblock In \bibinfo{booktitle}{\emph{Advances in Neural Information
  Processing Systems 31}}, \bibfield{editor}{\bibinfo{person}{S.~Bengio},
  \bibinfo{person}{H.~Wallach}, \bibinfo{person}{H.~Larochelle},
  \bibinfo{person}{K.~Grauman}, \bibinfo{person}{N.~Cesa-Bianchi}, {and}
  \bibinfo{person}{R.~Garnett}} (Eds.). \bibinfo{publisher}{Curran Associates,
  Inc.}, \bibinfo{pages}{1973--1973}.
\newblock


\bibitem[\protect\citeauthoryear{Cary, Das, Edelman, Giotis, Heimerl, Karlin,
  Kominers, Mathieu, and Schwarz}{Cary et~al\mbox{.}}{2014}]%
        {CDEG+14}
\bibfield{author}{\bibinfo{person}{M. Cary}, \bibinfo{person}{A. Das},
  \bibinfo{person}{B. Edelman}, \bibinfo{person}{I. Giotis},
  \bibinfo{person}{K. Heimerl}, \bibinfo{person}{A.~R. Karlin},
  \bibinfo{person}{S.~D. Kominers}, \bibinfo{person}{C. Mathieu}, {and}
  \bibinfo{person}{M. Schwarz}.} \bibinfo{year}{2014}\natexlab{}.
\newblock \showarticletitle{Convergence of Position Auctions Under Myopic
  Best-Response Dynamics}.
\newblock \bibinfo{journal}{\emph{ACM Trans. Econ. Comput.}}
  \bibinfo{volume}{2}, \bibinfo{number}{3}, Article \bibinfo{articleno}{9}
  (\bibinfo{date}{July} \bibinfo{year}{2014}), \bibinfo{numpages}{20}~pages.
\newblock
\showISSN{2167-8375}
\urldef\tempurl%
\url{https://doi.org/10.1145/2632226}
\showDOI{\tempurl}


\bibitem[\protect\citeauthoryear{Chen and Deng}{Chen and Deng}{2011}]%
        {CD11}
\bibfield{author}{\bibinfo{person}{N. Chen} {and} \bibinfo{person}{X. Deng}.}
  \bibinfo{year}{2011}\natexlab{}.
\newblock \showarticletitle{On Nash Dynamics of Matching Market Equilibria}.
\newblock  (\bibinfo{date}{03} \bibinfo{year}{2011}).
\newblock


\bibitem[\protect\citeauthoryear{Chen, Lu, and Lu}{Chen et~al\mbox{.}}{2019}]%
        {CLL19}
\bibfield{author}{\bibinfo{person}{P.-A. Chen}, \bibinfo{person}{C.-J. Lu},
  {and} \bibinfo{person}{Y.-S. Lu}.} \bibinfo{year}{2019}\natexlab{}.
\newblock \showarticletitle{An Alternating Algorithm for Finding Linear
  {A}rrow-{D}ebreu Market Equilibrium}.
\newblock \bibinfo{journal}{\emph{CoRR}}  \bibinfo{volume}{abs/1902.01754}
  (\bibinfo{year}{2019}).
\newblock
\urldef\tempurl%
\url{http://arxiv.org/abs/1902.01754}
\showURL{%
\tempurl}


\bibitem[\protect\citeauthoryear{Chen, Dai, Du, and Teng}{Chen
  et~al\mbox{.}}{2009}]%
        {chen2009settling}
\bibfield{author}{\bibinfo{person}{X. Chen}, \bibinfo{person}{D. Dai},
  \bibinfo{person}{Y. Du}, {and} \bibinfo{person}{S.-H. Teng}.}
  \bibinfo{year}{2009}\natexlab{}.
\newblock \showarticletitle{Settling the complexity of {A}rrow-{D}ebreu
  equilibria in markets with additively separable utilities}. In
  \bibinfo{booktitle}{\emph{Proc. of the 50th Ann. IEEE Symp. on Foundations of
  Computer Science}}. \bibinfo{pages}{273--282}.
\newblock


\bibitem[\protect\citeauthoryear{Cheung and Cole}{Cheung and Cole}{2014}]%
        {cheung2014amortized}
\bibfield{author}{\bibinfo{person}{Y.~K. Cheung} {and} \bibinfo{person}{R.
  Cole}.} \bibinfo{year}{2014}\natexlab{}.
\newblock \showarticletitle{Amortized analysis on asynchronous gradient
  descent}.
\newblock \bibinfo{journal}{\emph{arXiv preprint arXiv:1412.0159}}
  (\bibinfo{year}{2014}).
\newblock


\bibitem[\protect\citeauthoryear{Cheung, Cole, and Devanur}{Cheung
  et~al\mbox{.}}{2019a}]%
        {cheung2019tatonnement}
\bibfield{author}{\bibinfo{person}{Y.~K. Cheung}, \bibinfo{person}{R. Cole},
  {and} \bibinfo{person}{N.~R. Devanur}.} \bibinfo{year}{2019}\natexlab{a}.
\newblock \showarticletitle{Tatonnement beyond gross substitutes? Gradient
  descent to the rescue}.
\newblock \bibinfo{journal}{\emph{Games and Economic Behavior}}
  (\bibinfo{year}{2019}).
\newblock


\bibitem[\protect\citeauthoryear{Cheung, Cole, and Rastogi}{Cheung
  et~al\mbox{.}}{2012}]%
        {cheung2012tatonnement}
\bibfield{author}{\bibinfo{person}{Y.~K. Cheung}, \bibinfo{person}{R. Cole},
  {and} \bibinfo{person}{A. Rastogi}.} \bibinfo{year}{2012}\natexlab{}.
\newblock \showarticletitle{Tatonnement in ongoing markets of complementary
  goods}.
\newblock \bibinfo{journal}{\emph{arXiv preprint arXiv:1211.2268}}
  (\bibinfo{year}{2012}).
\newblock


\bibitem[\protect\citeauthoryear{Cheung, Cole, and Tao}{Cheung
  et~al\mbox{.}}{2018}]%
        {CCT18}
\bibfield{author}{\bibinfo{person}{Y.~K. Cheung}, \bibinfo{person}{R. Cole},
  {and} \bibinfo{person}{Y. Tao}.} \bibinfo{year}{2018}\natexlab{}.
\newblock \showarticletitle{Dynamics of Distributed Updating in {F}isher
  Markets}. In \bibinfo{booktitle}{\emph{Proc. of the 2018 ACM Conf. on
  Economics and Computation}}. \bibinfo{pages}{351--368}.
\newblock


\bibitem[\protect\citeauthoryear{Cheung, Hoefer, and Nakhe}{Cheung
  et~al\mbox{.}}{2019b}]%
        {CHN19}
\bibfield{author}{\bibinfo{person}{Y.~K. Cheung}, \bibinfo{person}{M. Hoefer},
  {and} \bibinfo{person}{P. Nakhe}.} \bibinfo{year}{2019}\natexlab{b}.
\newblock \showarticletitle{Tracing Equilibrium in Dynamic Markets via
  Distributed Adaptation}. In \bibinfo{booktitle}{\emph{Proc. of the 18th Int'l
  Conf. on Autonomous Agents and MultiAgent Systems}}.
  \bibinfo{pages}{1225--1233}.
\newblock
\showISBNx{978-1-4503-6309-9}


\bibitem[\protect\citeauthoryear{Codenotti, McCune, and Varadarajan}{Codenotti
  et~al\mbox{.}}{2005a}]%
        {codenotti2005market}
\bibfield{author}{\bibinfo{person}{B. Codenotti}, \bibinfo{person}{B. McCune},
  {and} \bibinfo{person}{K. Varadarajan}.} \bibinfo{year}{2005}\natexlab{a}.
\newblock \showarticletitle{Market equilibrium via the excess demand function}.
  In \bibinfo{booktitle}{\emph{Proc. of the 37th Ann. ACM Symp. on Theory of
  Computing}}. \bibinfo{pages}{74--83}.
\newblock


\bibitem[\protect\citeauthoryear{Codenotti, Pemmaraju, and
  Varadarajan}{Codenotti et~al\mbox{.}}{2005b}]%
        {codenotti2005polynomial}
\bibfield{author}{\bibinfo{person}{B. Codenotti}, \bibinfo{person}{S.
  Pemmaraju}, {and} \bibinfo{person}{K. Varadarajan}.}
  \bibinfo{year}{2005}\natexlab{b}.
\newblock \showarticletitle{On the polynomial time computation of equilibria
  for certain exchange economies}. In \bibinfo{booktitle}{\emph{Proc. of the
  16th Ann. ACM-SIAM Symp. on Discrete Algorithms}}. \bibinfo{pages}{72--81}.
\newblock


\bibitem[\protect\citeauthoryear{Codenotti, Saberi, Varadarajan, and
  Ye}{Codenotti et~al\mbox{.}}{2006}]%
        {codenotti2006leontief}
\bibfield{author}{\bibinfo{person}{B. Codenotti}, \bibinfo{person}{A. Saberi},
  \bibinfo{person}{K. Varadarajan}, {and} \bibinfo{person}{Y. Ye}.}
  \bibinfo{year}{2006}\natexlab{}.
\newblock \showarticletitle{Leontief economies encode nonzero sum two-player
  games}. In \bibinfo{booktitle}{\emph{Proc. of the 17h Ann. ACM-SIAM Symp. on
  Discrete Algorithm}}. \bibinfo{pages}{659--667}.
\newblock


\bibitem[\protect\citeauthoryear{Cole and Fleischer}{Cole and
  Fleischer}{2008}]%
        {cole2008fast}
\bibfield{author}{\bibinfo{person}{R. Cole} {and} \bibinfo{person}{L.
  Fleischer}.} \bibinfo{year}{2008}\natexlab{}.
\newblock \showarticletitle{Fast-converging tatonnement algorithms for one-time
  and ongoing market problems}. In \bibinfo{booktitle}{\emph{Proc. of the 40th
  Ann. ACM Symp. on Theory of Computing}}. \bibinfo{pages}{315--324}.
\newblock


\bibitem[\protect\citeauthoryear{Cornet}{Cornet}{1989}]%
        {cornet1989linear}
\bibfield{author}{\bibinfo{person}{B. Cornet}.}
  \bibinfo{year}{1989}\natexlab{}.
\newblock \showarticletitle{Linear exchange economies}.
\newblock \bibinfo{journal}{\emph{Cahier Eco-Math, Universit{\'e} de Paris}}
  \bibinfo{volume}{1} (\bibinfo{year}{1989}).
\newblock


\bibitem[\protect\citeauthoryear{Daskalakis, Deckelbaum, and Kim}{Daskalakis
  et~al\mbox{.}}{2015}]%
        {DDK15}
\bibfield{author}{\bibinfo{person}{C. Daskalakis}, \bibinfo{person}{A.
  Deckelbaum}, {and} \bibinfo{person}{A. Kim}.}
  \bibinfo{year}{2015}\natexlab{}.
\newblock \showarticletitle{Near-optimal no-regret algorithms for zero-sum
  games}.
\newblock \bibinfo{journal}{\emph{GEB}}  \bibinfo{volume}{92}
  (\bibinfo{year}{2015}), \bibinfo{pages}{327--348}.
\newblock


\bibitem[\protect\citeauthoryear{Daskalakis and Syrgkanis}{Daskalakis and
  Syrgkanis}{2016}]%
        {DS16}
\bibfield{author}{\bibinfo{person}{C. Daskalakis} {and} \bibinfo{person}{V.
  Syrgkanis}.} \bibinfo{year}{2016}\natexlab{}.
\newblock \showarticletitle{Learning in Auctions: Regret is Hard, Envy is
  Easy}. In \bibinfo{booktitle}{\emph{FOCS}}. \bibinfo{pages}{219--228}.
\newblock


\bibitem[\protect\citeauthoryear{Deng, Papadimitriou, and Safra}{Deng
  et~al\mbox{.}}{2003}]%
        {deng2003complexity}
\bibfield{author}{\bibinfo{person}{X. Deng}, \bibinfo{person}{C.
  Papadimitriou}, {and} \bibinfo{person}{S. Safra}.}
  \bibinfo{year}{2003}\natexlab{}.
\newblock \showarticletitle{On the complexity of price equilibria}.
\newblock \bibinfo{journal}{\emph{J. Comput. Syst. Sci.}} \bibinfo{volume}{67},
  \bibinfo{number}{2} (\bibinfo{year}{2003}), \bibinfo{pages}{311--324}.
\newblock


\bibitem[\protect\citeauthoryear{Devanur, Garg, and V{\'e}gh}{Devanur
  et~al\mbox{.}}{2016}]%
        {devanur2016rational}
\bibfield{author}{\bibinfo{person}{N.~R. Devanur}, \bibinfo{person}{J. Garg},
  {and} \bibinfo{person}{L.~A. V{\'e}gh}.} \bibinfo{year}{2016}\natexlab{}.
\newblock \showarticletitle{A rational convex program for linear
  {A}rrow-{D}ebreu markets}.
\newblock \bibinfo{journal}{\emph{ACM Trans. on Economics and Computation}}
  \bibinfo{volume}{5}, \bibinfo{number}{1} (\bibinfo{year}{2016}),
  \bibinfo{pages}{6}.
\newblock


\bibitem[\protect\citeauthoryear{Devanur and Kannan}{Devanur and
  Kannan}{2008}]%
        {devanurKannan2008market}
\bibfield{author}{\bibinfo{person}{N.~R. Devanur} {and} \bibinfo{person}{R.
  Kannan}.} \bibinfo{year}{2008}\natexlab{}.
\newblock \showarticletitle{Market equilibria in polynomial time for fixed
  number of goods or agents}. In \bibinfo{booktitle}{\emph{Proc. of the 49th
  Ann. IEEE Symp. on Foundations of Computer Science}}.
  \bibinfo{pages}{45--53}.
\newblock


\bibitem[\protect\citeauthoryear{Devanur, Papadimitriou, Saberi, and
  Vazirani}{Devanur et~al\mbox{.}}{2008}]%
        {devanur2008market}
\bibfield{author}{\bibinfo{person}{N.~R. Devanur}, \bibinfo{person}{C.~H.
  Papadimitriou}, \bibinfo{person}{A. Saberi}, {and} \bibinfo{person}{V.~V.
  Vazirani}.} \bibinfo{year}{2008}\natexlab{}.
\newblock \showarticletitle{Market equilibrium via a primal--dual algorithm for
  a convex program}.
\newblock \bibinfo{journal}{\emph{J. ACM}} \bibinfo{volume}{55},
  \bibinfo{number}{5} (\bibinfo{year}{2008}), \bibinfo{pages}{22}.
\newblock


\bibitem[\protect\citeauthoryear{Devanur and Vazirani}{Devanur and
  Vazirani}{2003}]%
        {devanur2003improved}
\bibfield{author}{\bibinfo{person}{N.~R. Devanur} {and} \bibinfo{person}{V.~V.
  Vazirani}.} \bibinfo{year}{2003}\natexlab{}.
\newblock \showarticletitle{An improved approximation scheme for computing
  {A}rrow-{D}ebreu prices for the linear case}. In
  \bibinfo{booktitle}{\emph{Int'l Conf. on Foundations of Software Technology
  and Theoretical Computer Science}}. \bibinfo{pages}{149--155}.
\newblock


\bibitem[\protect\citeauthoryear{Dütting and Kesselheim}{Dütting and
  Kesselheim}{2017}]%
        {DK17}
\bibfield{author}{\bibinfo{person}{P. Dütting} {and} \bibinfo{person}{T.
  Kesselheim}.} \bibinfo{year}{2017}\natexlab{}.
\newblock \showarticletitle{Best-Response Dynamics in Combinatorial Auctions
  with Item Bidding}. In \bibinfo{booktitle}{\emph{SODA}}.
  \bibinfo{pages}{521--533}.
\newblock


\bibitem[\protect\citeauthoryear{Duan, Garg, and Mehlhorn}{Duan
  et~al\mbox{.}}{2016}]%
        {duan2016improved}
\bibfield{author}{\bibinfo{person}{R. Duan}, \bibinfo{person}{J. Garg}, {and}
  \bibinfo{person}{K. Mehlhorn}.} \bibinfo{year}{2016}\natexlab{}.
\newblock \showarticletitle{An improved combinatorial polynomial algorithm for
  the linear {A}rrow-{D}ebreu market}. In \bibinfo{booktitle}{\emph{Proc. of
  the 27th annual ACM-SIAM Symp. on Discrete Algorithms}}.
  \bibinfo{pages}{90--106}.
\newblock


\bibitem[\protect\citeauthoryear{Duan and Mehlhorn}{Duan and Mehlhorn}{2015}]%
        {duan2015combinatorial}
\bibfield{author}{\bibinfo{person}{R. Duan} {and} \bibinfo{person}{K.
  Mehlhorn}.} \bibinfo{year}{2015}\natexlab{}.
\newblock \showarticletitle{A combinatorial polynomial algorithm for the linear
  {A}rrow-{D}ebreu market}.
\newblock \bibinfo{journal}{\emph{Information and Computation}}
  \bibinfo{volume}{243} (\bibinfo{year}{2015}), \bibinfo{pages}{112--132}.
\newblock


\bibitem[\protect\citeauthoryear{Echenique and Wierman}{Echenique and
  Wierman}{2011}]%
        {echenique2011finding}
\bibfield{author}{\bibinfo{person}{F. Echenique} {and} \bibinfo{person}{A.
  Wierman}.} \bibinfo{year}{2011}\natexlab{}.
\newblock \showarticletitle{Finding a {W}alrasian equilibrium is easy for a
  fixed number of agents}.
\newblock  (\bibinfo{year}{2011}).
\newblock


\bibitem[\protect\citeauthoryear{Eisenberg and Gale}{Eisenberg and
  Gale}{1959}]%
        {eisenberg1959consensus}
\bibfield{author}{\bibinfo{person}{E. Eisenberg} {and} \bibinfo{person}{D.
  Gale}.} \bibinfo{year}{1959}\natexlab{}.
\newblock \showarticletitle{Consensus of subjective probabilities: The
  pari-mutuel method}.
\newblock \bibinfo{journal}{\emph{The Annals of Mathematical Statistics}}
  \bibinfo{volume}{30}, \bibinfo{number}{1} (\bibinfo{year}{1959}),
  \bibinfo{pages}{165--168}.
\newblock


\bibitem[\protect\citeauthoryear{Feldman, Lai, and Zhang}{Feldman
  et~al\mbox{.}}{2005}]%
        {feldman2005price}
\bibfield{author}{\bibinfo{person}{M. Feldman}, \bibinfo{person}{K. Lai}, {and}
  \bibinfo{person}{L. Zhang}.} \bibinfo{year}{2005}\natexlab{}.
\newblock \showarticletitle{A price-anticipating resource allocation mechanism
  for distributed shared clusters}. In \bibinfo{booktitle}{\emph{Proc. of the
  6th ACM Conf. on Electronic Commerce}}. \bibinfo{pages}{127--136}.
\newblock


\bibitem[\protect\citeauthoryear{Fisher}{Fisher}{1983}]%
        {Fisher83}
\bibfield{author}{\bibinfo{person}{F.~M. Fisher}.}
  \bibinfo{year}{1983}\natexlab{}.
\newblock \bibinfo{booktitle}{\emph{Disequilibrium foundations of equilibrium
  economics}}.
\newblock \bibinfo{publisher}{Cambridge U. Press}.
\newblock


\bibitem[\protect\citeauthoryear{Fleischer, Garg, Kapoor, Khandekar, and
  Saberi}{Fleischer et~al\mbox{.}}{2008}]%
        {FGKKS08}
\bibfield{author}{\bibinfo{person}{L. Fleischer}, \bibinfo{person}{R. Garg},
  \bibinfo{person}{S. Kapoor}, \bibinfo{person}{R. Khandekar}, {and}
  \bibinfo{person}{A. Saberi}.} \bibinfo{year}{2008}\natexlab{}.
\newblock \showarticletitle{A Fast and Simple Algorithm for Computing Market
  Equilibria}. In \bibinfo{booktitle}{\emph{Proc. of the 4th Int'l Workshop on
  Internet and Network Economics}}. \bibinfo{pages}{19--30}.
\newblock


\bibitem[\protect\citeauthoryear{Florig}{Florig}{2004}]%
        {florig2004equilibrium}
\bibfield{author}{\bibinfo{person}{M. Florig}.}
  \bibinfo{year}{2004}\natexlab{}.
\newblock \showarticletitle{Equilibrium correspondence of linear exchange
  economies}.
\newblock \bibinfo{journal}{\emph{Journal of optimization theory and
  applications}} \bibinfo{volume}{120}, \bibinfo{number}{1}
  (\bibinfo{year}{2004}), \bibinfo{pages}{97--109}.
\newblock


\bibitem[\protect\citeauthoryear{Freund and Schapire}{Freund and
  Schapire}{1999}]%
        {FS99}
\bibfield{author}{\bibinfo{person}{Y. Freund} {and} \bibinfo{person}{R.~E
  Schapire}.} \bibinfo{year}{1999}\natexlab{}.
\newblock \showarticletitle{Adaptive game playing using multiplicative
  weights}.
\newblock \bibinfo{journal}{\emph{GEB}} \bibinfo{volume}{29},
  \bibinfo{number}{1-2} (\bibinfo{year}{1999}), \bibinfo{pages}{79--103}.
\newblock


\bibitem[\protect\citeauthoryear{Gale}{Gale}{1963}]%
        {Gale63}
\bibfield{author}{\bibinfo{person}{D. Gale}.} \bibinfo{year}{1963}\natexlab{}.
\newblock \showarticletitle{A note on global instability of competitive
  equilibrium}.
\newblock \bibinfo{journal}{\emph{Naval Research Logistics Quarterly}}
  \bibinfo{volume}{10}, \bibinfo{number}{1} (\bibinfo{year}{1963}),
  \bibinfo{pages}{81--87}.
\newblock


\bibitem[\protect\citeauthoryear{Gale}{Gale}{1976}]%
        {gale1976linear}
\bibfield{author}{\bibinfo{person}{D. Gale}.} \bibinfo{year}{1976}\natexlab{}.
\newblock \showarticletitle{The linear exchange model}.
\newblock \bibinfo{journal}{\emph{Journal of Mathematical Economics}}
  \bibinfo{volume}{3}, \bibinfo{number}{2} (\bibinfo{year}{1976}),
  \bibinfo{pages}{205--209}.
\newblock


\bibitem[\protect\citeauthoryear{Garg, Mehta, Sohoni, and Vazirani}{Garg
  et~al\mbox{.}}{2015}]%
        {garg2015complementary}
\bibfield{author}{\bibinfo{person}{J. Garg}, \bibinfo{person}{R. Mehta},
  \bibinfo{person}{M. Sohoni}, {and} \bibinfo{person}{V.~V. Vazirani}.}
  \bibinfo{year}{2015}\natexlab{}.
\newblock \showarticletitle{A complementary pivot algorithm for market
  equilibrium under separable, piecewise-linear concave utilities}.
\newblock \bibinfo{journal}{\emph{SIAM J. Comp.}} \bibinfo{volume}{44},
  \bibinfo{number}{6} (\bibinfo{year}{2015}), \bibinfo{pages}{1820--1847}.
\newblock


\bibitem[\protect\citeauthoryear{Garg and V{\'e}gh}{Garg and V{\'e}gh}{2019}]%
        {garg2019strongly}
\bibfield{author}{\bibinfo{person}{J. Garg} {and} \bibinfo{person}{L.~A.
  V{\'e}gh}.} \bibinfo{year}{2019}\natexlab{}.
\newblock \showarticletitle{A strongly polynomial algorithm for linear exchange
  markets}. In \bibinfo{booktitle}{\emph{Proc. of the 51st Ann. ACM Symp. on
  Theory of Computing}}. \bibinfo{pages}{54--65}.
\newblock


\bibitem[\protect\citeauthoryear{Garg and Kapoor}{Garg and Kapoor}{2006}]%
        {garg2006auction}
\bibfield{author}{\bibinfo{person}{R. Garg} {and} \bibinfo{person}{S. Kapoor}.}
  \bibinfo{year}{2006}\natexlab{}.
\newblock \showarticletitle{Auction algorithms for market equilibrium}.
\newblock \bibinfo{journal}{\emph{Math. Oper. Res.}} \bibinfo{volume}{31},
  \bibinfo{number}{4} (\bibinfo{year}{2006}), \bibinfo{pages}{714--729}.
\newblock


\bibitem[\protect\citeauthoryear{Hassidim, Kaplan, Mansour, and Nisan}{Hassidim
  et~al\mbox{.}}{2011}]%
        {HKMN11}
\bibfield{author}{\bibinfo{person}{A. Hassidim}, \bibinfo{person}{H. Kaplan},
  \bibinfo{person}{Y. Mansour}, {and} \bibinfo{person}{N. Nisan}.}
  \bibinfo{year}{2011}\natexlab{}.
\newblock \showarticletitle{Non-price equilibria in markets of discrete goods}.
  In \bibinfo{booktitle}{\emph{EC}}. \bibinfo{pages}{295--296}.
\newblock


\bibitem[\protect\citeauthoryear{Jain}{Jain}{2007}]%
        {Jain}
\bibfield{author}{\bibinfo{person}{K. Jain}.} \bibinfo{year}{2007}\natexlab{}.
\newblock \showarticletitle{A Polynomial Time Algorithm for Computing an
  {A}rrow-{D}ebreu Market Equilibrium for Linear Utilities}.
\newblock \bibinfo{journal}{\emph{SIAM J. Comput.}} \bibinfo{volume}{37},
  \bibinfo{number}{1} (\bibinfo{date}{April} \bibinfo{year}{2007}),
  \bibinfo{pages}{303--318}.
\newblock


\bibitem[\protect\citeauthoryear{Jain, Mahdian, and Saberi}{Jain
  et~al\mbox{.}}{2003}]%
        {jain2003approximating}
\bibfield{author}{\bibinfo{person}{K. Jain}, \bibinfo{person}{M. Mahdian},
  {and} \bibinfo{person}{A. Saberi}.} \bibinfo{year}{2003}\natexlab{}.
\newblock \showarticletitle{Approximating market equilibria}.
\newblock In \bibinfo{booktitle}{\emph{Approximation, Randomization, and
  Combinatorial Optimization. Algorithms and Techniques}}.
  \bibinfo{publisher}{Springer}, \bibinfo{pages}{98--108}.
\newblock


\bibitem[\protect\citeauthoryear{Kelly}{Kelly}{1997}]%
        {kelly1997charging}
\bibfield{author}{\bibinfo{person}{F. Kelly}.} \bibinfo{year}{1997}\natexlab{}.
\newblock \showarticletitle{Charging and rate control for elastic traffic}.
\newblock \bibinfo{journal}{\emph{Europ. Trans. on Telecommunications}}
  \bibinfo{volume}{8}, \bibinfo{number}{1} (\bibinfo{year}{1997}),
  \bibinfo{pages}{33--37}.
\newblock


\bibitem[\protect\citeauthoryear{Khalil}{Khalil}{2002}]%
        {khalil2002nonlinear}
\bibfield{author}{\bibinfo{person}{H.~K. Khalil}.}
  \bibinfo{year}{2002}\natexlab{}.
\newblock \showarticletitle{Nonlinear systems}.
\newblock \bibinfo{journal}{\emph{Upper Saddle River}} (\bibinfo{year}{2002}).
\newblock


\bibitem[\protect\citeauthoryear{Kleinberg, Piliouras, and Tardos}{Kleinberg
  et~al\mbox{.}}{2009}]%
        {KPT09}
\bibfield{author}{\bibinfo{person}{R. Kleinberg}, \bibinfo{person}{G.
  Piliouras}, {and} \bibinfo{person}{E. Tardos}.}
  \bibinfo{year}{2009}\natexlab{}.
\newblock \showarticletitle{Multiplicative updates outperform generic no-regret
  learning in congestion games}. In \bibinfo{booktitle}{\emph{STOC}}.
  \bibinfo{pages}{533--542}.
\newblock


\bibitem[\protect\citeauthoryear{Lucier and Borodin}{Lucier and
  Borodin}{2010}]%
        {lucier2010price}
\bibfield{author}{\bibinfo{person}{B. Lucier} {and} \bibinfo{person}{A.
  Borodin}.} \bibinfo{year}{2010}\natexlab{}.
\newblock \showarticletitle{Price of anarchy for greedy auctions}. In
  \bibinfo{booktitle}{\emph{SODA}}. \bibinfo{pages}{537--553}.
\newblock


\bibitem[\protect\citeauthoryear{Lykouris, Syrgkanis, and Tardos}{Lykouris
  et~al\mbox{.}}{2016}]%
        {LST16}
\bibfield{author}{\bibinfo{person}{T. Lykouris}, \bibinfo{person}{V.
  Syrgkanis}, {and} \bibinfo{person}{E. Tardos}.}
  \bibinfo{year}{2016}\natexlab{}.
\newblock \showarticletitle{Learning and Efficiency in Games with Dynamic
  Population}. In \bibinfo{booktitle}{\emph{SODA}}. \bibinfo{pages}{120--129}.
\newblock


\bibitem[\protect\citeauthoryear{Mehta, Panageas, and Piliouras}{Mehta
  et~al\mbox{.}}{2015}]%
        {MPP15}
\bibfield{author}{\bibinfo{person}{R. Mehta}, \bibinfo{person}{I. Panageas},
  {and} \bibinfo{person}{G. Piliouras}.} \bibinfo{year}{2015}\natexlab{}.
\newblock \showarticletitle{Natural selection as an inhibitor of genetic
  diversity: Multiplicative weights updates algorithm and a conjecture of
  haploid genetics [working paper abstract]}. In
  \bibinfo{booktitle}{\emph{ITCS}}. \bibinfo{pages}{73--73}.
\newblock


\bibitem[\protect\citeauthoryear{Mertens}{Mertens}{2003}]%
        {mertens2003limit}
\bibfield{author}{\bibinfo{person}{J.-F. Mertens}.}
  \bibinfo{year}{2003}\natexlab{}.
\newblock \showarticletitle{The limit-price mechanism}.
\newblock \bibinfo{journal}{\emph{Journal of Mathematical Economics}}
  \bibinfo{volume}{39}, \bibinfo{number}{5-6} (\bibinfo{year}{2003}),
  \bibinfo{pages}{433--528}.
\newblock


\bibitem[\protect\citeauthoryear{Nisan, Schapira, Valiant, and Zohar}{Nisan
  et~al\mbox{.}}{2011}]%
        {nisan2011best}
\bibfield{author}{\bibinfo{person}{N. Nisan}, \bibinfo{person}{M. Schapira},
  \bibinfo{person}{G. Valiant}, {and} \bibinfo{person}{A. Zohar}.}
  \bibinfo{year}{2011}\natexlab{}.
\newblock \showarticletitle{Best-response auctions}. In
  \bibinfo{booktitle}{\emph{EC}}. \bibinfo{pages}{351--360}.
\newblock


\bibitem[\protect\citeauthoryear{Panageas and Piliouras}{Panageas and
  Piliouras}{2016}]%
        {PP16}
\bibfield{author}{\bibinfo{person}{I. Panageas} {and} \bibinfo{person}{G.
  Piliouras}.} \bibinfo{year}{2016}\natexlab{}.
\newblock \showarticletitle{Average case performance of replicator dynamics in
  potential games via computing regions of attraction}. In
  \bibinfo{booktitle}{\emph{EC}}. \bibinfo{pages}{703--720}.
\newblock


\bibitem[\protect\citeauthoryear{Papadimitriou and Piliouras}{Papadimitriou and
  Piliouras}{2016}]%
        {PpadP16}
\bibfield{author}{\bibinfo{person}{C. Papadimitriou} {and} \bibinfo{person}{G.
  Piliouras}.} \bibinfo{year}{2016}\natexlab{}.
\newblock \showarticletitle{From Nash Equilibria to Chain Recurrent Sets:
  Solution Concepts and Topology}. In \bibinfo{booktitle}{\emph{ITCS}}.
  \bibinfo{pages}{227--235}.
\newblock


\bibitem[\protect\citeauthoryear{Roughgarden, Syrgkanis, and
  Tardos}{Roughgarden et~al\mbox{.}}{2017}]%
        {RST17}
\bibfield{author}{\bibinfo{person}{T. Roughgarden}, \bibinfo{person}{V.
  Syrgkanis}, {and} \bibinfo{person}{E. Tardos}.}
  \bibinfo{year}{2017}\natexlab{}.
\newblock \showarticletitle{The Price of Anarchy in Auctions}.
\newblock \bibinfo{journal}{\emph{J. Artif. Intell. Res.}}
  \bibinfo{volume}{59} (\bibinfo{year}{2017}), \bibinfo{pages}{59--101}.
\newblock


\bibitem[\protect\citeauthoryear{Scarf}{Scarf}{1960}]%
        {scarf1960some}
\bibfield{author}{\bibinfo{person}{H. Scarf}.} \bibinfo{year}{1960}\natexlab{}.
\newblock \showarticletitle{Some examples of global instability of the
  competitive equilibrium}.
\newblock \bibinfo{journal}{\emph{International Economic Review}}
  \bibinfo{volume}{1}, \bibinfo{number}{3} (\bibinfo{year}{1960}),
  \bibinfo{pages}{157--172}.
\newblock


\bibitem[\protect\citeauthoryear{Scarf}{Scarf}{1977}]%
        {scarf1977computation}
\bibfield{author}{\bibinfo{person}{H. Scarf}.} \bibinfo{year}{1977}\natexlab{}.
\newblock \bibinfo{booktitle}{\emph{The computation of equilibrium prices: an
  exposition}}.
\newblock \bibinfo{type}{{T}echnical {R}eport}. \bibinfo{institution}{Cowles
  Foundation for Research in Economics, Yale University}.
\newblock


\bibitem[\protect\citeauthoryear{Shapley and Shubik}{Shapley and
  Shubik}{1977}]%
        {shapley1977trade}
\bibfield{author}{\bibinfo{person}{L. Shapley} {and} \bibinfo{person}{M.
  Shubik}.} \bibinfo{year}{1977}\natexlab{}.
\newblock \showarticletitle{Trade using one commodity as a means of payment}.
\newblock \bibinfo{journal}{\emph{Journal of Political Economy}}
  \bibinfo{volume}{85}, \bibinfo{number}{5} (\bibinfo{year}{1977}),
  \bibinfo{pages}{937--968}.
\newblock


\bibitem[\protect\citeauthoryear{Shmyrev}{Shmyrev}{2009}]%
        {shmyrev2009algorithm}
\bibfield{author}{\bibinfo{person}{V.~I. Shmyrev}.}
  \bibinfo{year}{2009}\natexlab{}.
\newblock \showarticletitle{An algorithm for finding equilibrium in the linear
  exchange model with fixed budgets}.
\newblock \bibinfo{journal}{\emph{Journal of Applied and Industrial
  Mathematics}} \bibinfo{volume}{3}, \bibinfo{number}{4}
  (\bibinfo{year}{2009}), \bibinfo{pages}{505}.
\newblock


\bibitem[\protect\citeauthoryear{Walras}{Walras}{1896}]%
        {walras1896elements}
\bibfield{author}{\bibinfo{person}{L. Walras}.}
  \bibinfo{year}{1896}\natexlab{}.
\newblock \bibinfo{booktitle}{\emph{{\'E}l{\'e}ments d'{\'e}conomie politique
  pure, ou, Th{\'e}orie de la richesse sociale}}.
\newblock \bibinfo{publisher}{F. Rouge}.
\newblock


\bibitem[\protect\citeauthoryear{Wu and Zhang}{Wu and Zhang}{2007}]%
        {wu2007proportional}
\bibfield{author}{\bibinfo{person}{F. Wu} {and} \bibinfo{person}{L. Zhang}.}
  \bibinfo{year}{2007}\natexlab{}.
\newblock \showarticletitle{Proportional response dynamics leads to market
  equilibrium}. In \bibinfo{booktitle}{\emph{Proc. of the 39th Ann. ACM Symp.
  on Theory of Computing}}. \bibinfo{pages}{354--363}.
\newblock


\bibitem[\protect\citeauthoryear{Ye}{Ye}{2008}]%
        {ye2008path}
\bibfield{author}{\bibinfo{person}{Y. Ye}.} \bibinfo{year}{2008}\natexlab{}.
\newblock \showarticletitle{A path to the {A}rrow-{D}ebreu competitive market
  equilibrium}.
\newblock \bibinfo{journal}{\emph{Mathematical Programming}}
  \bibinfo{volume}{111}, \bibinfo{number}{1-2} (\bibinfo{year}{2008}),
  \bibinfo{pages}{315--348}.
\newblock


\bibitem[\protect\citeauthoryear{Zhang}{Zhang}{2011}]%
        {Zhang11}
\bibfield{author}{\bibinfo{person}{L. Zhang}.} \bibinfo{year}{2011}\natexlab{}.
\newblock \showarticletitle{Proportional response dynamics in the {F}isher
  market}.
\newblock \bibinfo{journal}{\emph{Theoretical Computer Science}}
  \bibinfo{volume}{412} (\bibinfo{year}{2011}), \bibinfo{pages}{2691--2698}.
\newblock


\end{thebibliography}


\end{document}